\documentclass[runningheads]{llncs}

\usepackage{pgfplots}
\pgfplotsset{compat=newest}
\usepackage{tikz}
\usetikzlibrary{arrows,matrix,positioning}
\usepackage[utf8]{inputenc}
\usepackage[english]{babel}
\usepackage[T1]{fontenc}
\usepackage{epsfig}
\usepackage{amsmath, amssymb, amsbsy}
\usepackage{mathdots}
\usepackage{xspace}
\usepackage[noend]{algpseudocode}
\usepackage{hyperref}
\usepackage[linesnumbered,ruled,vlined,titlenumbered]{algorithm2e}
\usepackage{algorithmicx}
\usepackage{color}
\usepackage{cite}
\usepackage{booktabs}
\usepackage{verbatim}
\usepackage{makecell}
\usepackage{lipsum}
\usepackage{enumitem}
\usepackage{colortbl}
\usepackage{theorem}
\usepackage{flushend}
\usepackage{nicefrac}
\usepackage{url}
\usepackage{lmodern}

\def\ve#1{{\mathchoice{\mbox{\boldmath$\displaystyle #1$}}%
              {\mbox{\boldmath$\textstyle #1$}}%
              {\mbox{\boldmath$\scriptstyle #1$}}%
              {\mbox{\boldmath$\scriptscriptstyle #1$}}}}

\renewcommand{\vec}[1]{\ensuremath{\mathbf{#1}}}

\renewcommand{\S}{\ensuremath{\mathbf{S}}}

\newcommand{\0}{\ensuremath{\mathbf{0}}}
\newcommand{\e}{\ensuremath{\mathbf{e}}}

\newcommand{\Mooremat}[3]{\ensuremath{\mathcal{M}_{#1,#2}\left( #3 \right)}}
\newcommand{\MoormatExplicit}[3]{
	\begin{pmatrix}
		#1_{1} & #1_{2} & \dots& #1_{#3}\\
		#1_{1}^{q} & #1_{2}^{q} & \dots& #1_{#3}^{q}\\
		\vdots &\vdots&\ddots& \vdots\\
		#1_{1}^{q^{#2-1}} & #1_{2}^{q^{#2-1}} & \dots& #1_{#3}^{q^{#2-1}}\\
	\end{pmatrix}}

\newcommand{\MoormatExplicitMat}[4]{
	\begin{pmatrix}
          #1_{1,1} & #1_{1,2} & \dots& #1_{1,#3}\\
          #1_{2,1} & #1_{2,2} & \dots& #1_{2,#3}\\
          \vdots &\vdots&\ddots& \vdots\\
          #1_{#4,1} & #1_{#4,2} & \dots& #1_{#4,#3}\\
          #1_{1,1}^{q} & #1_{1,2}^{q} & \dots& #1_{1,#3}^{q}\\
          #1_{2,1}^{q} & #1_{2,2}^{q} & \dots& #1_{2,#3}^{q}\\
                    \vdots &\vdots&\ddots& \vdots\\
		#1_{#4,1}^{q^{#2-1}} & #1_{#4,2}^{q^{#2-1}} & \dots& #1_{#4,#3}^{q^{#2-1}}\\
	\end{pmatrix}}

\newcommand{\Fq}{\ensuremath{\mathbb{F}_q}}
\newcommand{\Fqm}{\ensuremath{\mathbb{F}_{q^m}}}
\newcommand{\Fqmu}{\ensuremath{\mathbb{F}_{q^{mu}}}}

\DeclareMathOperator{\extsmallfield}{ext}
\DeclareMathOperator{\rank}{rk}
\DeclareMathOperator{\Tr}{Tr}
\newcommand{\nkhalffrac}{\left\lfloor \frac{n-k}{2}\right\rfloor}

\newcommand{\RkError}{\ensuremath{\varphi}}
\newcommand{\dimZ}{\ensuremath{\zeta}}
\newcommand{\MooreZ}{\ensuremath{\tilde{\vec{Z}}}}

\newcommand{\tpub}{\ensuremath{t_{\mathsf{pub}}}}
\newcommand{\kpub}{\ensuremath{\mathbf{k}_{\mathsf{pub}}}}
\newcommand{\mycode}[1]{\ensuremath{\mathcal{#1}}}
\newcommand{\fontmetric}[1]{\mathsf{#1}}

\newcommand{\codelinearRank}[1]{\ensuremath{[#1]_q^\fontmetric{R}}}

\newcommand{\Gabcode}[2]{\ensuremath{\mathcal{G}(#1,#2)}}
\newcommand{\IntGabcode}[1]{\ensuremath{\mathcal{IG}(#1)}}

\newcommand{\numbRowErasures}{\varrho}
\newcommand{\numbColErasures}{\gamma}

\newcommand{\y}{\vec{y}}
\renewcommand{\a}{\vec{a}}

\newcommand{\g}{\vec{g}}
\newcommand{\m}{\vec{m}}
\renewcommand{\c}{\vec{c}}
\renewcommand{\e}{\vec{e}}
\newcommand{\K}{\vec{K}}
\newcommand{\M}{\vec{M}}
\renewcommand{\P}{\vec{P}}
\newcommand{\GG}{\vec{G}_\mathcal{G}}
\newcommand{\Trmum}{\Tr}
\renewcommand{\k}{\vec{k}}
\newcommand{\s}{\vec{s}}
\newcommand{\z}{\vec{z}}
\newcommand{\x}{\vec{x}}
\newcommand{\X}{\vec{X}}
\newcommand{\Y}{\vec{Y}}
\newcommand{\A}{\vec{A}}
\newcommand{\E}{\vec{E}}
\newcommand{\G}{\vec{G}}
\newcommand{\Pinv}{\P^{-1}}

\newcommand{\SystemName}{\textsf{LIGA}}

\newcommand{\keyGen}{\ensuremath{\mathsf{KeyGen}}}
\newcommand{\enc}{\ensuremath{\mathsf{Encrypt}}}
\newcommand{\dec}{\ensuremath{\mathsf{Decrypt}}}
\newcommand{\sys}{\ensuremath{\mathsf{\Pi}^{\text{PKE}}}}
\newcommand{\syskem}{\ensuremath{\mathsf{\Pi}^{\text{KEM}}}}
\newcommand{\encap}{\ensuremath{\mathsf{Encaps}}}
\newcommand{\decap}{\ensuremath{\mathsf{Decaps}}}
\newcommand{\secLevel}{\ensuremath{\lambda}}

\newcommand{\cipherset}{\hat{\mathcal{C}}}

\newcommand{\pk}{\ensuremath{\mathsf{pk}}}
\newcommand{\sk}{\ensuremath{\mathsf{sk}}}
\newcommand{\ct}{\ensuremath{\mathsf{ct}}}

\newcommand{\sample}{\xleftarrow{\$}}

\newcommand{\Adv}{\ensuremath{\mathsf{Adv}}}

\newcommand{\qbinomial}[2]{\genfrac{[}{]}{0pt}{}{#1}{#2}_q}
\newcommand{\qbinomialField}[3]{\genfrac{[}{]}{0pt}{}{#1}{#2}_{#3}}

\newcommand{\Code}{\mathcal{C}}

\newcommand{\trrank}{t}

\newcommand{\dRmin}{\mathrm{d}_{\mathrm{rk},\mathrm{min}}}

\newcommand{\quadbinom}[2]{\ensuremath{
		{#1
			\brack
			#2}_q
}}

\makeatletter
\newcommand{\removelatexerror}{\let\@latex@error\@gobble}
\makeatother

\begin{document}

\title{\SystemName:
A Cryptosystem Based on the Hardness of Rank-Metric List and Interleaved Decoding}
\titlerunning{\SystemName: A Rank-Metric Code-Based Cryptosystem}

\author{Julian Renner\inst{1} \and Sven Puchinger\inst{2} \and Antonia Wachter-Zeh\inst{1}}

\institute{  Technical University of Munich (TUM), Munich, Germany\\
  \email{\{julian.renner, antonia.wachter-zeh\}@tum.de}\thanks{The work of J. Renner and A. Wachter-Zeh was supported by the European Research Council (ERC) under the European Union's Horizon 2020 research and innovation programme (grant agreement No 801434).}
  \and
Technical University of Denmark (DTU), Lyngby, Denmark \\\email{svepu@dtu.dk}\thanks{Sven Puchinger has received funding from the European Union's
Horizon 2020 research and innovation program under the Marie Sklodowska-Curie grant
agreement no.~713683 (COFUNDfellowsDTU).}
}

\date{\today}

\maketitle

\begin{abstract}

We propose the new rank-metric code-based cryptosystem {\SystemName} which is based on the hardness of \underline{l}ist decoding and \underline{i}nterleaved decoding of \underline{Ga}bidulin codes.
{\SystemName} is an improved variant of the \emph{Faure--Loidreau} (FL) system, which was broken in a structural attack by \emph{Gaborit, Otmani, and {Tal\'e Kalachi}} (GOT, 2018).
We keep the FL encryption and decryption algorithms, but modify the insecure key generation algorithm.
Our crucial observation is that the GOT attack is equivalent to decoding an interleaved Gabidulin code.
The new key generation algorithm constructs public keys for which all polynomial-time interleaved decoders fail---hence {\SystemName} resists the GOT attack.
We also prove that the public-key encryption version of {\SystemName} is IND-CPA secure in the standard model and the KEM version is IND-CCA2 secure in the random oracle model, both under hardness assumptions of formally defined problems related to list decoding and interleaved decoding of Gabidulin codes.
We propose and analyze various exponential-time attacks on these problems,
calculate their work factors, and compare the resulting parameters to NIST proposals.
The strengths of {\SystemName} are short ciphertext sizes and (relatively) small key sizes. Further, {\SystemName} guarantees correct decryption and has no decryption failure rate.
It is \emph{not} based on hiding the structure of a code.
Since there are efficient and constant-time algorithms for encoding and decoding Gabidulin codes, timing attacks on the encryption and decryption algorithms can be easily prevented.

\end{abstract}

\section{Introduction}

Public-key cryptography is the foundation for establishing secure communication between multiple parties.
Traditional public-key algorithms such as RSA are based on the hardness of factoring large numbers or the discrete logarithm problem, but can be attacked in polynomial time once a capable quantum computer exists. Code-based public-key cryptosystems are considered to be post-quantum secure, but compared to RSA or elliptic curve cryptography their crucial drawback is the significantly larger key size.
Recently, the \emph{National Institute of Standards and Technology} (NIST) has initiated a standardization progress for post-quantum secure public-key algorithms \cite{NIST}. The currently being evaluated Round 2 of the competition consists of 17 code-based and lattice-based public-key encryption algorithms.  
The NIST competition and its systems attract a lot of attention and show the importance of designing post-quantum secure public-key encryption algorithms.

The \emph{Faure--Loidreau} (FL) code-based cryptosystem~\cite{faure2006new,LoidreauHabitilation-RankMetric_2007} is based on the problem of reconstructing linearized polynomials and can be seen as linearized equivalent of the (broken) Augot--Finiasz cryptosystem \cite{AugotFiniasz-PKC-PolyReconstruction_2003}.
While the Augot--Finiasz cryptosystem is closely connected to (list) decoding Reed--Solomon codes, the 
FL cryptosystem is connected to (list) decoding Gabidulin codes, a special class of rank-metric codes~\cite{Gabidulin_TheoryOfCodes_1985}.
In contrast to McEliece-type (or Niederreiter-type) cryptosystems, where the public key is a \emph{matrix}, in the FL system, the public key is only a \emph{vector}, resulting in a much smaller key size.
At the time when the FL cryptosystem was designed, it was only \emph{conjectured} that Gabidulin codes cannot be list decoded efficiently. As this was \emph{proven} in the last years for many families of Gabidulin codes \cite{Wachterzeh_BoundsListDecodingRankMetric_IEEE-IT_2013,RavivWachterzeh_GabidulinBounds_journal,ListDec_RankMetric_2019}, the FL system could be a very promising post-quantum secure public-key cryptosystem.
However, the recent structural attack by Gaborit, Otmani and Tal\'e Kalachi \cite{Gaborit2018-FL-Attack} can recover an alternative public key in cubic time complexity.

{In this paper, a new system is presented which is based on the original FL system, and therefore relies on the proven hardness of list decoding Gabidulin codes, but makes the attack from \cite{Gaborit2018-FL-Attack} impossible.}
Our contributions are as follows. First, a new coding-theoretic interpretation of the original FL system is given and an alternative decryption algorithm is proposed. Second, we show that the public key can be seen as a corrupted codeword of an \emph{interleaved Gabidulin code}. We prove that the failure condition of the GOT attack~\cite{Gaborit2018-FL-Attack} on the public key is equivalent to the failure condition of decoding the public key as a corrupted interleaved Gabidulin codeword.
This observation enables us to design a new code-based public-key encryption scheme, as well as a corresponding key encapsulation mechanism (KEM), based on the hardness of list and interleaved decoding Gabidulin codes: \SystemName.
In {\SystemName}, we choose the public key in a way that the corresponding interleaved decoder is guaranteed to fail, and thus, the system is secured against the attack from~\cite{Gaborit2018-FL-Attack}.
We also prove that the public-key encryption version of {\SystemName} is IND-CPA secure in the standard model and the KEM version is IND-CCA2 secure in the random oracle model, both under hardness assumptions on problems related to list and interleaved decoding of Gabidulin codes. We analyze possible (exponential-time) attacks on these hard problems, provide sets of parameters for {\SystemName}, and compare them amongst others to NIST proposals (RQC, ROLLO, BIKE, McEliece).

The structure of this paper is as follows. In Section~\ref{sec:preliminaries}, the notation is introduced and definitions are given. 
In Section~\ref{sec:fl_system}, the key generation of the original FL system is shown and a new coding-theoretic interpretation of the ciphertext and the public key is derived. After summarizing the attack from~\cite{Gaborit2018-FL-Attack}, we prove its equivalence to decoding the public key as an interleaved Gabidulin code. 
Based on this equivalence, the new system {\SystemName} is proposed in Section~\ref{sec:liga} and
its IND-CPA and IND-CCA2 security are proven in Section~\ref{sec:semantic}.
A security analysis of our system is given in Section~\ref{sec:sec_analysis}. In Section~\ref{sec:parameters}, example parameters for security levels $128$, $192$, and $256$ bit are proposed and compared to the NIST proposals RQC \cite{melchor2019rqc}, ROLLO \cite{melchor2019rollo}, BIKE \cite{aragon2019bike}, ClassicMcEliece \cite{bernstein2019mceliece} and Loidreau's McEliece-like system from \cite{Loidreau2017-NewRankMetricBased}. Conclusions are given in Section~\ref{sec:conclusion}.

Parts of these results have been presented at the \emph{IEEE International Symposium on Information Theory 2018} \cite{wachterzeh2018repairing}. The content of this journal paper contains various new results that were not shown in \cite{wachterzeh2018repairing}. In particular, in this paper,\\[-3ex]
\begin{itemize}
	\item we generalize {\SystemName}'s Key Generation algorithm, i.e., the choice of the $\vec{z}_i$'s (the interleaved errors in the public key) is more flexible now (in \cite{wachterzeh2018repairing}, $\vec{z}_1 = \vec{z}_2 = \dots = \vec{z}_u$),
	\item we present a KEM/DEM version of {\SystemName},
	\item the encryption and decryption complexity is analyzed,
	\item we show a new way to realize the decryption by error-erasure decoding,
	\item we identify formal problems in the rank metric on which the security of {\SystemName} relies and prove the IND-CPA/CCA2 security of the KEM/DEM version under the assumption that some of these problems are hard,
	\item we analyze new exponential-time attacks on these problems,
	\item we update the choice of parameters.
\end{itemize}

It is important to note that all these results go well beyond what is known about and what has been analyzed for the original FL system.

\section{Preliminaries}\label{sec:preliminaries}
\subsection{Notations}
Let $q$ be a power of a prime and let $\Fq$ denote the finite field of order $q$. Then, $\Fqm$ and $\Fqmu$ denote extension fields of $\Fq$ of order $q^m$ and $q^{mu}$, respectively. 
We use $\Fq^{m \times n}$ to denote the set of all $m\times n$ matrices over $\Fq$ and $\Fqm^n =\Fqm^{1 \times n}$ for the set of all row vectors of length $n$ over $\Fqm$. 
Further, we use another field extension $\mathbb{F}_{q^{mu}}$ with $u>1$.
Thus, $\Fq \subseteq \Fqm \subseteq \Fqmu$.

For a field $\mathbb{F}$, the vector space that is spanned by $\vec{v}_1,\hdots,\vec{v}_l \in \mathbb{F}^n$ is denoted by
\begin{equation*}
\langle\vec{v}_1,\hdots,\vec{v}_l\rangle_{\mathbb{F}} := \Bigg\{\sum_{i=1}^{l}a_i\vec{v}_i \, : \, \ a_i \in \mathbb{F} \Bigg\}.  
\end{equation*}

Denote the set of integers $[a,b] = \{i: a \leq i \leq b\}$. Rows and columns of $m\times n$-matrices are indexed by $1,\dots, m$ and $1,\dots, n$, where $A_{i,j}$ is the element in the $i$-th row and $j$-th column of the matrix $\vec{A}$. Further,
\begin{equation*}
  \vec{A}_{[a,b]} :=
  \begin{pmatrix}
    A_{1,a} & \hdots & A_{1,b} \\
    \vdots & \ddots & \vdots \\
    A_{m,a} & \hdots & A_{m,b} \\
  \end{pmatrix}.
\end{equation*}
  By $\rank_q(\vec{A})$ and $\rank_{q^m}(\vec{A})$, we denote the rank of a matrix $\vec{A}$ over $\Fq$, respectively $\Fqm$.
Let $(\gamma_1,\gamma_2,\dots,\gamma_{u})$ be an ordered basis of $\Fqmu$ over $\Fqm$. By utilizing the vector space isomorphism $\Fqmu \cong \Fqm^u$, we can relate each vector $\mathbf a \in \Fqmu^n$ to a matrix $\mathbf A \in \Fqm^{u \times n}$ according to
\begin{align*}
\extsmallfield_{\boldsymbol{\gamma}}:\Fqm^{n} &\rightarrow \Fq^{m \times n}\label{eq:mapping_smallfield}\\
  \vec{a} = (a_1,\hdots,a_n) &\mapsto \vec{A} =
                               \begin{pmatrix}
                                 A_{1,1} & \hdots & A_{1,n} \\
                                 \vdots & \ddots & \vdots \\
                                 A_{m,1} &  \hdots & A_{m,n} \\
                               \end{pmatrix},
\end{align*}
where $\boldsymbol{\gamma} = (\gamma_1,\gamma_2,\dots,\gamma_{u})$ and
\begin{equation*}
\quad a_j = \sum_{i=1}^{m} A_{i,j} \gamma_i, \quad \forall j \in [1,n].
\end{equation*}
The trace operator of a vector $\vec{a}\in\Fqmu$ to $\Fqm$ is defined by
\begin{align*}
  \Trmum:\Fqmu^{n} &\rightarrow \Fqm^{n}\\
  \vec{a} = (a_1,a_2,\hdots,a_n) &\mapsto \Bigg( \sum_{i=0}^{m-1} a_1^{q^i}, \sum_{i=0}^{m-1} a_2^{q^i},\hdots, \sum_{i=0}^{m-1} a_n^{q^i} \Bigg).
\end{align*}
A dual basis $(\gamma_1^*,\gamma_2^*,\dots,\gamma_{u}^*)$ to  $(\gamma_1,\gamma_2,\dots,\gamma_{u})$ is a basis that fulfills
\begin{equation*}
  \Trmum(\gamma_i \gamma_j^{*}) = 
  \begin{cases}
    1 ~\text{if $i=j$} \\
    0 ~\text{else}
  \end{cases},
\end{equation*}
where $i,j \in [1,u]$. Note that a dual basis always exists.

Denote by $\Mooremat{s}{q}{\vec{a}} \in \Fqm^{s \times n}$ the $s \times n$ Moore matrix for a vector $\vec{a} = (a_1,a_2,\dots,a_n) \in \Fqm^n$, i.e.,
\begin{equation*}
\Mooremat{s}{q}{\vec{a}} := \MoormatExplicit{a}{s}{n}.
\end{equation*}
If $a_1, a_2,\dots$, $a_{n}\in \Fqm$ are linearly independent over $\Fq$, then $\rank_{q^m}(\Mooremat{s}{q}{\vec{a}})=\min\{s,n\}$, cf.~\cite[Lemma 3.15]{Lidl-Niederreiter:FF1996}. This definition can also be extended to matrices by 
\begin{equation*}
\Mooremat{s}{q}{\vec{A}} := \MoormatExplicitMat{A}{s}{n}{l},
\end{equation*}
where $\vec{A} \in \Fqm^{l \times n}$.

The Gaussian binomial coefficient is denoted by
\begin{equation*}
  \quadbinom{s}{r} :=
  \begin{cases}
    \dfrac{ (1-q^s)(1-q^{s-1})\cdots(1-q^{s-r+1})}{(1-q)(1-q^2)\cdots(1-q^{r})} &\text{ for $r \leq s$} \\
    0 &\text{ for $r> s$},
    \end{cases}
\end{equation*}
where $s$ and $r$ are non-negative integers.

Let $\mathcal{X}$ be a set. When $x$ is drawn uniformly at random from the set $\mathcal{X}$, we denote it by $x \sample \mathcal{X}$. Further, be $x \gets y$ we mean that we assign $y$ to $x$.

\subsection{Rank-Metric Codes and Gabidulin Codes}\label{ssec:preliminaries_rank-metric_codes}
The \textit{rank norm} $\rank_q(\vec{a})$ is the rank of the matrix representation $\vec{A}\in \Fq^{m \times n}$ over $\mathbb{F}_{q}$. 
The rank distance between \vec{a} and \vec{b} is the rank of the difference of the two matrix representations, i.e.,
\begin{equation*}
d_{\textup{R}}(\vec{a},\vec{b}) := \rank_q(\vec{a}-\vec{b}) = \rank_q(\vec{A}-\vec{B}).
\end{equation*}
An $\codelinearRank{n,k,d}$ code \mycode{C} over $\Fqm$ is a linear rank-metric code, i.e., it is a linear subspace of $\Fqm^n$ of dimension $k$ and minimum rank distance
\begin{equation*}
d := \min_{\substack{{\vec{a},\vec{b}} \in \mycode{C}\\ \vec{a} \neq \vec{b}}}
\big\lbrace d_{\fontmetric{R}}(\vec{a},\vec{b}) = \rank_q(\vec{a}-\vec{b}) \big\rbrace. 
\end{equation*}
For linear codes with $n \leq m$, the Singleton-like upper bound \cite{Delsarte_1978,Gabidulin_TheoryOfCodes_1985} implies that $d \leq n-k+1$.
If $d=n-k+1$, the code is called a \emph{maximum rank distance} (MRD) code.

Gabidulin codes \cite{Gabidulin_TheoryOfCodes_1985} are a special class of rank-metric codes and can be defined by their generator matrices.
\begin{definition}[Gabidulin Code \cite{Gabidulin_TheoryOfCodes_1985}]
	A linear $\Gabcode{n}{k}$ code over $\Fqm$ of length $n \leq m$ 
	and dimension $k$ is defined by its $k \times n$ generator matrix
	\begin{equation*}
	\mathbf G_{\mycode{G}} = \Mooremat{k}{q}{g_1,g_2,\dots,g_n},
	\end{equation*}
	where $\vec{g}=(g_1,g_2, \dots, g_{n}) \in \Fqm^n$ and $\rank_q(\vec{g}) = n$. 
\end{definition}
In~\cite{Gabidulin_TheoryOfCodes_1985}, it is shown that Gabidulin codes are MRD codes, i.e., $d=n-k+1$.

For a short description on decoding of Gabidulin codes, denote by $\vec{C}_\mathcal{G}\in \Fq^{m \times n}$ the transmitted codeword (i.e., the matrix representation of $\vec{c}_\mathcal{G}\in \Fqm^n$) of a $\Gabcode{n}{k}$ code that is corrupted by an additive error $\mathbf E \in \Fq^{m \times n}$. At the receiver side, only the received matrix $\mathbf R \in \Fq^{m \times n}$, where $\mathbf R = \vec{C}_\mathcal{G}+ \mathbf E$, is known. 
The channel might provide additional side information in the form of erasures:
\begin{itemize}
	\item $\numbRowErasures$ {row erasures} (in \cite{silva_rank_metric_approach} called "deviations") and
	\item $\numbColErasures$ {column erasures} (in \cite{silva_rank_metric_approach} called "erasures"),
\end{itemize}  
such that the received matrix can be decomposed into 
\begin{equation}\label{eq:decomp_errrorerasures}
\vec{R}
=\vec{C}_\mathcal{G}+ \underbrace{\vec{A}^{(R)} \vec{B}^{(R)} + \vec{A}^{(C)} \vec{B}^{(C)} + \vec{E}}_{= \,\vec{E}_\mathrm{total}},
\end{equation}
where $\vec{A}^{(R)} \in \Fq^{m \times \numbRowErasures}$, $\vec{B}^{(R)} \in \Fq^{\numbRowErasures \times n}$, 
$\vec{A}^{(C)} \in \Fq^{m \times \numbColErasures}$, $\vec{B}^{(C)} \in \Fq^{\numbColErasures\times n}$ are full-rank matrices, respectively, and $\vec{E}^{(E)} \in \Fq^{m \times n}$ is a matrix of rank $t$.
The decoder knows $\vec{R}$ and additionally $\vec{A}^{(R)}$ and $\vec{B}^{(C)}$. 
Further, $t$ denotes the number of errors without side information. 
The rank-metric error-erasure decoding algorithms from \cite{GabidulinPilipchuck_ErrorErasureRankCodes_2008,silva_rank_metric_approach,WachterzehZeh-ListUniqueErrorErasureInterpolationInterleavedGabidulin_DCC2014}
can then reconstruct $\vec{c}_\mathcal{G} \in \Gabcode{n}{k}$ with asymptotic complexity $\mathcal O(n^2)$ operations over $\Fqm$, or in sub-quadratic complexity using the fast operations described in \cite{PuchingerWachterzeh-ISIT2016,puchinger2018fast}, if
\begin{equation}\label{eq:errorerasurecond}
2t + \numbRowErasures + \numbColErasures \leq d-1 = n-k
\end{equation}
is fulfilled.

\subsection{{Interleaved Rank-Metric Codes}}\label{subsec:int-rank-codes}
Interleaved Gabidulin Codes are a code class for which efficient decoders are known that are able to correct w.h.p. random errors of rank larger than $\lfloor\frac{d-1}{2}\rfloor$.
\begin{definition} [Interleaved Gabidulin Codes~\cite{Loidreau_Overbeck_Interleaved_2006}]
A linear (vertically, homogeneous) interleaved Gabidulin code $\IntGabcode{u;n,k}$ over $\Fqm$ of length $n \leq m$, dimension $k \leq n$, and interleaving order $u$ is defined by
\begin{equation*}
\IntGabcode{u;n,k} :=
\left\lbrace
\begin{pmatrix}
\vec{c}_{\mycode{G}}^{(1)}\\
\vec{c}_{\mycode{G}}^{(2)}\\
\vdots\\
\vec{c}_{\mycode{G}}^{(u)}
\end{pmatrix}
: \vec{c}_{\mycode{G}}^{(i)} \in \Gabcode{n}{k} , \forall i \in [1,u]
\right\rbrace.
\end{equation*}
\end{definition}
As a short-term notation, we also speak about a $u$-interleaved Gabidulin code.
When considering random errors of rank weight $t$, the code $\IntGabcode{u;n,k}$ can be decoded uniquely with high probability up to 
$w \leq \lfloor \frac{u}{u+1}(n-k)\rfloor$ errors 
\footnote{
	{In this setting, an ``error of weight $w$'' is a matrix in $\Fqmu^n$ with $\Fq$-rank equal to $w$.}},
 cf.~\cite{Loidreau_Overbeck_Interleaved_2006,Sidorenko2011SkewFeedback,WachterzehZeh-ListUniqueErrorErasureInterpolationInterleavedGabidulin_DCC2014}.
However, it is well-known that there are many error patterns for which the known efficient decoders fail.
In fact, we can explicitly construct a large class of such errors as shown in the following lemma.

\begin{lemma}[Interleaved Decoding~{\cite{Loidreau_Overbeck_Interleaved_2006}, \cite{Sidorenko2011SkewFeedback}, \cite[p.~64]{Wachterzeh_DecodingBlockConvolutionalRankMetric_2013}}]\label{lem:interleaved-fail}
	Let $\vec{c}_i=\vec{x}_i \cdot \vec{G}_{\mathcal{G}}$.
	All known\footnote{i.e., the algorithms in \cite{Loidreau_Overbeck_Interleaved_2006,Sidorenko2011SkewFeedback}, and \cite[p.~64]{Wachterzeh_DecodingBlockConvolutionalRankMetric_2013}.} efficient decoders for $\IntGabcode{u;n,k}$ codes fail to correct an error $\vec{z}\in \Fqmu^n$ with $\vec{z} = \sum_{i=1}^{u} \vec{z}_i\gamma_i^*$ and $\rank_q(\vec{z})=w$ if 
	\begin{equation*}
	\rank_{q^m}
	\begin{pmatrix}
	\Mooremat{n-w-1}{q}{\vec{g}}\\
	\Mooremat{n-k-w}{q}{\vec{c}_1+\vec{z}_1}\\
	\Mooremat{n-k-w}{q}{\vec{c}_2+\vec{z}_2}\\
	\vdots\\
	\Mooremat{n-k-w}{q}{\vec{c}_u+\vec{z}_u}\\
	\end{pmatrix}
	< n-1.
	\end{equation*}
\end{lemma}
It is widely conjectured that there cannot be a decoder that decodes the error patterns of Lemma~\ref{lem:interleaved-fail} \emph{uniquely}. 
Decoding these failing error patterns has been subject to intensive research since the Loidreau--Overbeck decoder \cite{Loidreau_Overbeck_Interleaved_2006} was found in 2006. In the Hamming metric, the equivalent problem for Reed--Solomon codes has been studied since 1997 \cite{krachkovsky1997decoding} and more than a dozen papers have dealt with decoding algorithms for these codes. None of these papers was able to give a polynomial-time decoding algorithm for the cases of Lemma~\ref{lem:interleaved-fail}. 
It seems that all unique decoders \emph{have to fail} for the error patterns of Lemma~\ref{lem:interleaved-fail} since for these cases, there is no unique decision, i.e., more then one interleaved codeword lies in the ball of radius $w$ around the received word.

\section{Key Generation in the  Original Faure--Loidreau System}\label{sec:fl_system}
In this section, we recall the key generation algorithm of the original FL cryptosystem, we give a coding-theoretic interpretation of the original public key, and analyze the structural attack from \cite{Gaborit2018-FL-Attack}.
\subsection{The Original Algorithm}

Let $q,m,n,k,u,w,\tpub$ be positive integers that fulfill the restrictions given in Table~\ref{tab:parameters} and are publicly known.
In the following, we consider the three finite fields $\Fq$, $\Fqm$, and $\Fqmu$, which are extension fields of each other, i.e.:
\begin{equation*}
\Fq \subseteq \Fqm \subseteq \Fqmu.
\end{equation*}
\begin{table*}
\caption{Summary of the Parameters}
\renewcommand{\arraystretch}{1.6} % to make the height of each row equal
\begin{center}
\begin{tabular}{c|l|l}
Name & Use & Restriction \\
\hline
$q$ & small field size & prime power \\
$m$ & extension degree & $1 \leq m$ \\
$n$ & code length & $n \leq m$ \\
$k$ & code dimension & $k < n$ \\
$u$ & extension degree & $2 \leq u < k$ \\
$w$ & error weight in public key & $\max \left\{n-k-\frac{k-u}{u-1}, \left\lfloor\frac{n-k}{2} \right\rfloor+1 \right\} \leq w < \frac{u}{u+2} (n-k)$ \\
$\tpub$ & error weight in ciphertext & $\tpub = \left\lfloor \frac{n-k-w}{2} \right\rfloor$ \\
$\dimZ$ & \makecell[l]{$\Fqm$-dimension of error vector \\ in the public key }& $\dimZ< \frac{w}{n-k-w}\ $  {and $\ \zeta q^{\zeta w-m} \leq \tfrac{1}{2}$}                                        
\end{tabular}
\end{center}
\label{tab:parameters}
\end{table*}
The original FL key generation is shown in Algorithm~\ref{alg:KeyGen}.
 \begin{algorithm}
 \DontPrintSemicolon
 \KwIn{Parameters $q,m,n,k,u,w$ as in Table~\ref{tab:parameters}}
 \KwOut{Secret key $\sk$, public key $\pk$}
 $\vec{g} \sample \{ \a \in \Fqm^n \,:\, \rank_q(\a) = n \} $\\
 $\vec{x} \sample \{ \a \in \Fqmu^k \,:\, \dim (\langle a_{k-u+1},\dots, a_k \rangle_{\Fqm} ) = u  \}$  \\
 $\vec{s} \sample \{ \a \in \Fqmu^w \,:\,\rank_{q} (\a) = w \}$ \label{line:keygen_s}\\
 $\vec{P} \sample \{ \A \in \Fq^{n\times n} : \rank_{q}(\A) = n  \}  $  \\
 $\mathbf G_{\mycode{G}} \gets \Mooremat{k}{q}{\vec{g}}$ \\
 $\vec{z} \gets (\vec{s} \ | \ \0) \cdot \vec{P}^{-1}$ \\
 $\kpub \gets \vec{x} \cdot \mathbf G_{\mycode{G}} + \vec{z}$ \label{line:encrypt_keygen} \\
 $\tpub \gets \left\lfloor \frac{n-w-k}{2} \right \rfloor$\\
 \Return{$\sk = (\vec{x},\vec{P}_{[w+1,n]})$, $\pk = (\vec{g}, \kpub)$}
 \caption{Original FL Key Generation}
 \label{alg:KeyGen}
 \end{algorithm}

\subsection{Coding-Theoretic Interpretation of the Original Public Key}\label{ssec:kpub_corrupted_interleaved_codeword}
The public key $\kpub$ of the FL system is a corrupted codeword of a $u$-interleaved Gabidulin code.
To our knowledge, this connection between the public key and interleaved Gabidulin codes has not been known before.
This interpretation is central to this paper and will be used in Section~\ref{sec:repair} to define the public key of {\SystemName} such that is not vulnerable against the attacks from \cite{Gaborit2018-FL-Attack} and described in Section~\ref{sec:eff_attack}.

\begin{theorem}\label{thm:public_key_is_corrupted_interleaved_codeword}
Fix a basis $\vec{\gamma}$ of $\Fqmu$ over $\Fqm$. Let $\vec{\gamma}^*$ be a dual basis to $\vec{\gamma}$ and write $\kpub = \sum_{i=1}^{u} \kpub^{(i)}\gamma_i^*$. Then,
\begin{equation}
\begin{pmatrix}
\kpub^{(1)}\\
\kpub^{(2)}\\
\vdots\\
\kpub^{(u)}
\end{pmatrix}
=
\begin{pmatrix}
\vec{c}_{\mycode{G}}^{(1)}\\
\vec{c}_{\mycode{G}}^{(2)}\\
\vdots\\
\vec{c}_{\mycode{G}}^{(u)}
\end{pmatrix}
+\!
\begin{pmatrix}
\vec{z}_1\\
\vec{z}_2\\
\vdots\\
\vec{z}_u
\end{pmatrix}, \label{eq:kpub_as_corrupted_interleaved_codeword}
\end{equation}
where the $\vec{c}_{\mycode{G}}^{(i)} \in \Fqm^n$ are codewords of the Gabidulin code $\mathcal{G}(n,k)$ with generator matrix $\GG$ and the $\vec{z}_i \in \Fqm^n$ are obtained from the vector $\vec{z} \in \Fqmu^n$ by $\vec{z} = \sum_{i=1}^{u} \vec{z}_i\gamma_i^*$.
\end{theorem}

\begin{proof}
Recall the definition of the public key
\begin{equation*}
\kpub = \vec{x} \cdot \mathbf G_{\mycode{G}} + \vec{z},
\end{equation*}
where $\vec{x} \in \Fqmu^k$, $\vec{G}_{\mycode{G}} \in \Fqm^{k\times n}$ is the generator matrix of a $\Gabcode{n}{k}$ code, and $\vec{z} \in \Fqmu^n$ with $\rank_q(\vec{z})=w$.
Let $\vec{x} = \sum_{i=1}^{u} \vec{x}_i\gamma_i^*$, where the $\vec{x}_i$ have coefficients in $\Fqm$.

Then, we obtain the following representation of the public key $\kpub$ as a $u \times n$ matrix in $\Fqm$
\begin{equation*}
\begin{pmatrix}
\kpub^{(1)}\\
\kpub^{(2)}\\
\vdots\\
\kpub^{(u)}
\end{pmatrix}
=
\begin{pmatrix}
\vec{x}_1\\
\vec{x}_2\\
\vdots\\
\vec{x}_u
\end{pmatrix}\!\cdot \vec{G}_{\mycode{G}}
+\!
\begin{pmatrix}
\vec{z}_1\\
\vec{z}_2\\
\vdots\\
\vec{z}_u
\end{pmatrix}
=
\begin{pmatrix}
\vec{x}_1\cdot \vec{G}_{\mycode{G}}\\
\vec{x}_2\cdot \vec{G}_{\mycode{G}}\\
\vdots\\
\vec{x}_u\cdot \vec{G}_{\mycode{G}}
\end{pmatrix}
+\!
\begin{pmatrix}
\vec{z}_1\\
\vec{z}_2\\
\vdots\\
\vec{z}_u
\end{pmatrix}.
\end{equation*}
Since $\vec{x}_i \cdot \vec{G}_{\mycode{G}}$ is a codeword of a $\Gabcode{n}{k}$ code, $\forall i \in [1,u]$, the matrix representation of $\kpub$ can be seen as a codeword from an $\IntGabcode{u;n,k}$ code, corrupted by an error.  
\qed
\end{proof}

Note that the error $(\z_1^\top,\dots,\z_u^\top)^\top$ in \eqref{eq:kpub_as_corrupted_interleaved_codeword} has $\Fq$-rank at most $w$ due to the structure of $\z = (\s \mid \vec{0}) \Pinv$.

\subsection{Efficient Key Recovery of the Original FL Key}\label{sec:eff_attack}
{The attack by \emph{Gaborit, Otmani and Tal\'e Kalachi} (GOT) on the original FL system in~\cite{Gaborit2018-FL-Attack} (see Algorithm~\ref{alg:got_attack} below) is an efficient structural attack which computes a valid private key of the FL system in cubic time when the public key fulfills certain conditions.} We recall this attack in the following and derive an alternative, equally powerful, attack based on interleaved decoding the public key, utilizing the observation of the previous subsection. We prove that the failure conditions of both attacks are equivalent.
The interleaved decoding attack does not have any advantage in terms of cryptanalysis compared to \cite{Gaborit2018-FL-Attack}, but enables us to exactly predict for which public keys both attacks work and for which the attacks fail.

\subsubsection{GOT Attack}
The key recovery in the GOT attack (Algorithm~\ref{alg:got_attack}) succeeds under the conditions of the following theorem.
 \begin{algorithm}
 \DontPrintSemicolon
 \KwIn{Public key $(\vec{g},\kpub)$}
 \KwOut{Private key $(\vec{x},\vec{P})$}
Choose $\gamma_1,\dots,\gamma_u $ to be a basis of $\Fqmu$ over $\Fqm$ \\
$\kpub^{(i)} \gets\Trmum(\gamma_i \kpub)$ for all $i=1,\hdots,u$ \\
$\mathbf G_{\mycode{G}} \gets \Mooremat{k}{q}{\vec{g}}$ \\
Pick at random a non-zero vector $\widetilde{\vec{h}} \in \Fqm^n$ such that
  \begin{equation*}
    \mathcal{M}_{n-w-k,q}
    \begin{pmatrix}
        \vec{G}_{\mycode{G}}\\
        \kpub^{(1)} \\
        \vdots \\
        \kpub^{(u)} \\
      \end{pmatrix}
      \cdot
    \widetilde{\vec{h}}^{\text{T}} = \vec{0}.
    \end{equation*} \\
Choose $\vec{P} \in \Fq^{n\times n}$ and $\vec{h}^{\prime}\in \Fqm^{n-w}$ such that
$    \widetilde{\vec{h}} \big(\vec{P}^{-1} \big)^{\text{T}} = (\ \vec{0} \ | \ \vec{h}^{\prime})$
\\
Choose $\vec{x}$ such that
$    \vec{x} \vec{G}_{\mycode{G}} \vec{P}^{\prime} = \kpub \vec{P}^{\prime}$,
  where $\vec{P}^{\prime}=\vec{P}_{[w+1,n]} \in \Fq^{n \times (n-w)}$ \\
\Return{$(\vec{x},\vec{P})$}
 \caption{GOT Attack}
 \label{alg:got_attack}
 \end{algorithm}

\begin{theorem}[GOT Attack {\cite[Thm.~1]{Gaborit2018-FL-Attack}}]\label{thm:GOT-attack-success}
	Let $\gamma_1,\dots,\gamma_u \in \Fqmu$ be a basis of $\Fqmu$ over $\Fqm$ and let $\vec{z}_i = \Trmum(\gamma_i \vec{z})$, for $i=1,\dots u$.

	If the matrix $\vec{Z} \in \Fqm^{u \times n}$ with $\vec{z}_1,\dots,\vec{z}_u$ as rows, satisfies
	\begin{equation*}
	\rank_{q^m} (\Mooremat{n-k-w}{q}{\vec{Z}}) = w,
	\end{equation*}
	then $(\vec{x}, \vec{z})$ can be recovered from $(\vec{G}_{\mycode{G}},\kpub)$ with $\mathcal{O}(n^3)$ operations in $\Fqmu$ by using Algorithm~\ref{alg:got_attack}.
\end{theorem}

If the key is generated by Algorithm~\ref{alg:KeyGen}, the GOT attack breaks the original FL system \emph{with high probability}.

\subsubsection{Interleaved Decoding Attack}\label{subsec:int-attack}
Recall from Theorem~\ref{thm:public_key_is_corrupted_interleaved_codeword} that the public key $\kpub$ is a corrupted interleaved codeword.
Based on this observation we will derive a structural attack on the original FL system to which we refer as \emph{Interleaved Decoding Attack} in the following. We prove that interleaved decoding and the GOT attack \emph{fail} (i.e, do not provide any information) for the \emph{same public keys}.
The idea is to decode $\kpub$ in an interleaved Gabidulin code. 
Since $w \leq  \frac{u}{u+1}(n-k)$, such a decoder will return $\vec{x}$ with high probability, but fail in certain cases, see Section~\ref{subsec:int-rank-codes}.

Since $\rank_{q^m}(\Mooremat{n-w-1}{q}{\vec{g}})=n-w-1$, the interleaved decoder fails if (compare Lemma~\ref{lem:interleaved-fail}):
\begin{equation}\label{eq:fail-interleaved}
\rank _{q^m} \big(\MooreZ \big) := \RkError< w,
\end{equation}	
where
\begin{equation}\label{eq:matrix-repair}
  \MooreZ = 
  \begin{pmatrix}
\Mooremat{n-k-w}{q}{\vec{z}_1}\\
\Mooremat{n-k-w}{q}{\vec{z}_2}\\
\vdots\\
\Mooremat{n-k-w}{q}{\vec{z}_u}
\end{pmatrix}.
\end{equation}

\subsubsection{Equivalence of GOT Attack and Interleaved Decoding Attack}

In the following, we prove that the failure condition of the GOT Attack is equivalent to the condition that decoding $\kpub$ in an interleaved Gabidulin code fails.
\begin{theorem}\label{thm:equiv-int-got}
	The GOT Attack from \cite{Gaborit2018-FL-Attack} fails if and only if the Interleaved Decoding Attack fails. In particular, both fail if \eqref{eq:fail-interleaved} holds.
\end{theorem}
\begin{proof}
Rewrite the matrix from Theorem~\ref{thm:GOT-attack-success} as
\begin{align}
\Mooremat{n-w-k}{q}{\vec{Z}}	=
\begin{pmatrix}
\vec{z}_1\\
\vdots\\
\vec{z}_u\\
\vec{z}_1^{q}\\
\vdots\\
\vec{z}_u^{q}\\
\vdots\\
\vec{z}_1^{q^{n-w-k-1}}\\
\vdots\\
\vec{z}_u^{q^{n-w-k-1}}\\
\end{pmatrix}
  \label{eq:thm1-rewrite}
\end{align}
and the matrix from equation \eqref{eq:matrix-repair} as
\begin{align}
\MooreZ
=
\begin{pmatrix}
\vec{z}_1\\
\vdots\\
\vec{z}_1^{q^{n-w-k-1}}\\
\vec{z}_2\\
\vdots\\
\vec{z}_2^{q^{n-w-k-1}}\\
\vdots\\
\vec{z}_u\\
\vdots\\
\vec{z}_u^{q^{n-w-k-1}}\\
\end{pmatrix}.
\label{eq:lemm2-rewrite}
\end{align}
Since the matrix in~\eqref{eq:thm1-rewrite} and in~\eqref{eq:lemm2-rewrite} only differ in row permutations, they are row-space equivalent, implying that they have the same rank. Further, the rank of the matrix in~\eqref{eq:lemm2-rewrite}  cannot become larger than~$w$ (since any vector in the right kernel of this matrix has rank weight at least $n-w$ \cite[Algorithm 3.2.1]{Overbeck_Diss_InterleveadGab}). 
Thus, the failures of Theorem~\ref{thm:GOT-attack-success} and Lemma~\ref{lem:interleaved-fail} are equivalent. \qed
\end{proof}

{In the next section, we will exploit the observation of Theorem~\ref{thm:equiv-int-got}, i.e., we propose a new key generation algorithm that avoids public keys that can be efficiently decoded by an interleaved decoder, thereby rendering the GOT attack useless.}

\section{The New System \SystemName}\label{sec:liga}
In this section, we propose a public-key code-based encryption scheme $\sys = (\keyGen,\enc,\dec)$ called $\SystemName$. The system is based on the original FL system~\cite{faure2006new}, where we keep both the original encryption and decryption algorithm, but replace the insecure key-generation algorithm. Further, we present a KEM-DEM version of $\SystemName$ denoted by $\syskem = (\keyGen,\encap,\decap)$. 

Later, in Section~\ref{sec:semantic}, we will analyze the security of the system. We single out problems from coding theory and we prove that the encryption version is IND-CPA secure and the KEM-DEM version is IND-CCA2 secure under the assumption that the stated problems are hard. Furthermore, we study new and known attacks on these problems and show that they all run in exponential time (see~Section~\ref{sec:sec_analysis}).

\subsection{The New Key Generation Algorithm}\label{sec:repair}
We introduce a new key generation algorithm that is based on choosing $\vec{z} = \sum_{i=1}^{u} \vec{z}_i\gamma^*_i$ in a way that $\RkError < w$, where $\RkError$ is the rank of the interleaved Moore matrix of the errors $\vec{z}_i$ in the public key, see \eqref{eq:matrix-repair}. Based on the dimension of the span of the $\vec{z}_i$, we will upper bound $\RkError$ in the following Theorem~\ref{thm:rank-S_mat}. Recall that when $\RkError < w$, the GOT attack~\cite{Gaborit2018-FL-Attack} and interleaved decoding of the public key fail, see Theorem~\ref{thm:equiv-int-got}. In this case, retrieving any knowledge about the private key from the public key requires to solve Problem~\ref{pro:restr-int-search-rsd} (defined later), which basically corresponds to decoding the interleaved codeword when error patterns occur for which all known decoders fail.

\begin{theorem}\label{thm:rank-S_mat}
  Let $\dim(\langle\vec{z}_1,\hdots,\vec{z}_u \rangle_{q^m}) = \dimZ$. Then
  \begin{equation}
    \label{eq:rank-S_mat}
    \RkError = \rank _{q^m} \big(\MooreZ \big) \leq \min \{\dimZ(n-k-w),w\}.
  \end{equation}
\end{theorem}
\begin{proof}
  The dimension of $\langle\vec{z}_1,\hdots,\vec{z}_u \rangle_{q^m}$ implies that at most $\dimZ(n-k-w)$ rows of $\MooreZ$ are linearly independent over $\Fqm$, meaning that $\RkError \leq \dimZ(n-k-w)$.
  The definition of $\vec{z}=(\vec{s} \ | \ \0) \cdot \vec{P}^{-1}$ leads to

  \begin{align}
\RkError &= \rank_{q^m}(\MooreZ)   \nonumber\\
         &= \rank_{q^m}
\begin{pmatrix}
           \begin{pmatrix}
    \begin{array}{c | c}
      \Mooremat{n-k-w}{q}{\vec{s}_1} & \vec{0}\\
      \vdots & \vdots\\
      \Mooremat{n-k-w}{q}{\vec{s}_u} & \vec{0}\\
    \end{array}
  \end{pmatrix}\vec{P}^{-1}
  \end{pmatrix} \nonumber\\ 
         &= \rank_{q^m}
    \begin{pmatrix}
    \begin{array}{c}
      \Mooremat{n-k-w}{q}{\vec{s}_1}\\
      \vdots\\
      \Mooremat{n-k-w}{q}{\vec{s}_u}\\
    \end{array}
  \end{pmatrix}\nonumber \\
         &\leq w, \nonumber
  \end{align}
where the last inequality holds since $\vec{s}_1,\hdots,\vec{s}_u$ are vectors of length $w$. \qed
\end{proof}

We propose the following modification to Line~\ref{line:keygen_s} of the \textbf{Key Generation}, depending on the parameter $\zeta$:
{\RestyleAlgo{plain}
\begin{algorithm}
\DontPrintSemicolon
{
\setcounter{AlgoLine}{2}
$\mathcal{A} \sample \left\{ \text{subspace } \mathcal{U} \subseteq \Fqm^w \, : \, \dim \mathcal{U} = \zeta, \, \mathcal{U} \text{ has a basis of full-$\Fq$-rank elements} \right\}$ \\[2ex]
\setcounter{AlgoLine}{2}
\SetNlSty{textbf}{}{'}
$\begin{pmatrix}
\s_1 \\
\vdots \\
\s_u
\end{pmatrix} \sample \left\{ \begin{pmatrix}
\s_1' \\
\vdots \\
\s_u'
\end{pmatrix} \, : \, \langle\s_1',\dots,\s_u'\rangle_{\Fqm} = \mathcal{A}, \ \rank_{q}(\s_i') = w, \, \forall \, i \right\}$ \\
}
\label{alg:s_construction}
\end{algorithm}
}

Clearly, $\dim(\langle\vec{z}_1,\hdots,\vec{z}_u \rangle_{q^m}) = \dimZ$ in this case.
{To avoid that the GOT attack~\cite{Gaborit2018-FL-Attack} runs in polynomial time, Theorem~\ref{thm:rank-S_mat} implies that the parameter $\dimZ$ must always be chosen such that $\dimZ < \frac{w}{n-k-w}$.
In Section~\ref{sec:sec_analysis}, we will discuss several further exponential-time attacks on {\SystemName}.
Some of these attacks have a work factor depending on $\zeta$, which must be considered in the parameter design.

Furthermore, the condition $\rank_{q}(\s_i') = w$ ensures that $\rank_{q}(\z_i)=w$, i.e., as large as possible for a given subspace $\mathcal{A}$.
This choice maximizes the work factor of generic decoding attacks on the rows of the public key (seen as a received word of an interleaved Gabidulin code), see Section~\ref{sec:sec_analysis}.

The restriction of the choice of $\mathcal{A}$ to subspaces that contain a basis of full-$\Fq$-rank codewords is to ensure that the set from which we sample in Line~3' is non-empty. Hence, the key generation always works.
}

Compared to the choice of $\vec{z}$ in Line~\ref{line:keygen_s} of the original Key Generation algorithm, we restrict the choice of $\z$, but we will see in Section~\ref{sec:sec_analysis} that there are still enough possibilities for $\vec{z}$ to prevent an efficient naive brute-force attack.

{Appendix~\ref{app:practical_key_generation} contains a more detailed discussion on how to realize Lines $3$ and $3'$ in practice.}

\subsection{The Public Key Encryption Version}
The new key generation algorithm $\keyGen$, the encryption algorithm $\enc$ and the decryption algorithm $\dec$ are shown in Algorithm~\ref{alg:modKeyGen}, Algorithm~\ref{alg:Encrypt} and Algorithm~\ref{alg:Decrypt}, respectively. Compared to original key generation algorithm, the algorithm {\keyGen} has one more input parameter $\dimZ$ (cf.~Section~\ref{sec:repair}).
 \begin{algorithm}
 \DontPrintSemicolon
 \KwIn{Parameters $q,m,n,k,u,w,\dimZ$ as in Table~\ref{tab:parameters}}
 \KwOut{Private key $\sk$, public key $\pk$}
 $\vec{g} \sample \{ \a \in \Fqm^n \,:\, \rank_q(\a) = n \} $\\
 $\vec{x} \sample \{ \a \in \Fqmu^k \,:\, \dim (\langle a_{k-u+1},\dots, a_k \rangle_{\Fqm} ) = u  \}$  \\
{
\setcounter{AlgoLine}{2}
$\mathcal{A} \sample \left\{ \text{subspace } \mathcal{U} \subseteq \Fqm^w \, : \, \dim \mathcal{U} = \zeta, \, \mathcal{U} \text{ has a basis of full-$\Fq$-rank elements} \right\}$ \\
\setcounter{AlgoLine}{2}
\SetNlSty{textbf}{}{'}
$\begin{pmatrix}
\s_1 \\
\vdots \\
\s_u
\end{pmatrix} \sample \left\{ \begin{pmatrix}
\s_1' \\
\vdots \\
\s_u'
\end{pmatrix} \, : \, \langle\s_1',\dots,\s_u'\rangle_{\Fqm} = \mathcal{A}, \, \rank_{q}(\s_i') = w \, \forall \, i \right\}$ \\
}
\SetNlSty{textbf}{}{}
 $\vec{s} \gets \sum_{i=1}^{u}\vec{s}_i\gamma_i^{*}$ \\
 $\vec{P} \sample \{ \A \in \Fq^{n\times n} : \rank_{q}(\A) = n  \}  $  \\
 $\mathbf G_{\mycode{G}} \gets \Mooremat{k}{q}{\vec{g}}$ \\ 
 $\vec{z} \gets (\vec{s} \ | \ \0) \cdot \vec{P}^{-1}$ \\
 $\kpub \gets \vec{x} \cdot \mathbf G_{\mycode{G}} + \vec{z}$ \label{line:encrypt_keygen} \\
 $\tpub \gets \left\lfloor \frac{n-w-k}{2} \right \rfloor$\\
 \Return{$\sk = (\vec{x},\vec{P}_{[w+1,n]})$, $\pk = (\vec{g}, \kpub)$}
 \caption{$\keyGen\,(\cdot)$}
 \label{alg:modKeyGen}
 \end{algorithm}

 \begin{algorithm}
 \KwIn{Plaintext $\vec{m} \in \Fqm^{k-u}$, public key $\pk = (\vec{g},\kpub)$, randomness $\theta$}
 \KwOut{Ciphertext $\vec{c}$}
 $\alpha \sample \Fqmu \setminus \{0\}$ using $\theta$ \\
 $\vec{e} \sample \{ \a \in \Fqm^n : \rank_q(\a) = \tpub \}$ using $\theta$ \\
 $\mathbf G_{\mycode{G}} \gets \Mooremat{k}{q}{\vec{g}}$ \\
 \Return{$\vec{c} \gets (\vec{m},\vec{0}_u)\cdot\mathbf G_{\mycode{G}} + \Trmum(\alpha \kpub) +\vec{e}.$\label{line:encrypt_encoding}}
 \caption{$\enc\, (\cdot)$}
 \label{alg:Encrypt}
 \end{algorithm}

 \begin{algorithm}
 \KwIn{Ciphertext $\vec{c}\in \Fqm^n$, private key $\sk = (\vec{x},\vec{P}_{[w+1,n]})$}
 \KwOut{Plaintext $\vec{m}\in \Fqm^{k-u}$}
 $\c' \gets$ $\vec{c}\vec{P}_{[w+1,n]}$ \label{line:decrypt_cP} \\
 $\mycode{\mathcal{G}}' \gets $ Gabidulin code generated by $\vec{G}_{\mycode{G}}\vec{P}_{[w+1,n]}$ \\
 $\vec{m}^{\prime} \gets$ decode $\vec{c}'$ in $\mycode{\mathcal{G}}'$ \label{line:decrypt_decoding} \\
 $\alpha \gets \sum_{i=k-u+1}^{k}m_i^{\prime}x_i^*$ \label{line:decrypt_alpha} \\
 \Return{$\vec{m} \gets (\vec{m}^{\prime}-\Trmum(\alpha\vec{x}))_{[1,k-u]}$} \label{line:decrypt_m}
 \caption{$\dec\,(\cdot)$}
 \label{alg:Decrypt}
 \end{algorithm}

 The proposed system has no decryption failures as proven in the following theorem.
\begin{theorem}[Correctness~\cite{faure2006new}]
Algorithm~\ref{alg:Decrypt} returns the correct plaintext $\vec{m}$.
\end{theorem}

\begin{proof}
Line~\ref{line:decrypt_cP} computes
\begin{equation*}
\vec{c}\vec{P} = (\vec{m} + \Trmum(\alpha\vec{x}))\vec{G}_{\mycode{G}}\vec{P} + (\Trmum(\alpha \vec{s}) | \vec{0}) + \vec{e}\vec{P},
\end{equation*}
whose last $n-w$ columns are given by
\begin{equation*}
\vec{c}^{\prime} = (\vec{m} + \Trmum(\alpha\vec{x}))\vec{G}^{\prime} + \vec{e}^{\prime},
\end{equation*}
where $\vec{G}^{\prime}:= \vec{G}_{\mycode{G}}\vec{P}_{[w+1,n]} \in \Fqm^{k \times (n-w)}$ and $\vec{e}^{\prime}:= \vec{e}\vec{P}_{[w+1,n]}$. By decoding in $\mycode{\mathcal{G}}'$, we thus obtain the vector
\begin{align*}
\vec{m}' = \vec{m} + \Trmum(\alpha\vec{x}).
\end{align*}
Since the last $u$ positions of the plaintext $\m$ are zero (i.e., $m_i=0$ for $i=k-u+1,\dots,k$), we get $\alpha = \sum_{i=k-u+1}^{k}m_i^{\prime}x_i^*$, where $\{x_{k-u+1}^*,\dots,x_k^*\}$ is a dual basis to $\{x_{k-u+1},\dots,x_k\}$. 
As we know $\alpha$ and $\x$, we can compute the plaintext~$\m$.
\qed
\end{proof}

\begin{remark}
{Steps~1 to~3 of Algorithm~\ref{alg:Decrypt} can be interpreted as an error-erasure decoder of a Gabidulin code.
As this observation may have advantages, especially for implementations, we present this connection formally in Appendix~\ref{app:decryption_as_error-erasure_decoding}.}
\end{remark}

A SageMath v8.8~\cite{stein_sagemath} implementation of the public key encryption version of {\SystemName} can be downloaded from \url{https://bitbucket.org/julianrenner/liga_pke}. The purpose of the source code is to clarify the shown algorithms but \emph{not} to provide a secure and efficient implementation.

\subsection{KEM/DEM Version $\sys$ and $\syskem$}
In~\cite{hofheinz2017modular}, generic transformations of IND-CPA secure public key encryptions into IND-CCA2 secure KEMs are proposed. In the following, we apply one of the transformations directly to $\sys$ to obtain $\syskem$. Later, in Section~\ref{ssec:semantic_security}, we will prove that $\sys$ fulfills the requirements such that the applied transformation is secure.

Let $\mathcal{G}$, $\mathcal{H}$ and $\mathcal{K}$ be hash functions, where $\mathcal{G}\neq\mathcal{H}$. In Algorithm~\ref{alg:Encap} and Algorithm~\ref{alg:Decap}, we show the encapsulation and decapsulation algorithms of the KEM $\syskem = (\keyGen,\encap,\decap)$. The algorithm $\keyGen$ remains Algorithm~\ref{alg:modKeyGen}. 

 \begin{algorithm}
 \KwIn{Public key $\pk = (\vec{g},\kpub)$}
 \KwOut{Ciphertext $(\vec{c},\vec{d})$, shared key $K$}
 $\m \sample \Fqm^{k-u}$ (this will serve as shared key) \\
 $\theta \gets \mathcal{G}(\m)$\\
 $\c\gets \enc(\m,\pk,\theta)$ \\
 $K\gets \mathcal{K}(\m,\c)$ \\
 $ \vec{d} \gets \mathcal{H}(\c,\vec{d}) $ \\
 \Return{$(\vec{c},\vec{d}),K$}
 \caption{$\encap\, (\cdot)$}
 \label{alg:Encap}
 \end{algorithm}

 \begin{algorithm}
 \KwIn{Ciphertext $(\vec{c},\vec{d})$, private key $\sk = (\vec{x},\vec{P}_{[w+1,n]})$}
 \KwOut{Shared key $K$}
 $\m' \gets \dec(\c,\sk)$ \\
 $\theta' \gets \mathcal{G}(\m')$ \\
 $\c' \gets \enc(\m',\pk,\theta)$ \\
 \eIf{$\c \neq \c'$ or $\vec{d} \neq \mathcal{H}(\m')$}
 {$K \gets \bot$}
 {$K \gets \mathcal{K}(\m',\c)$}
 \Return{$K$}
 \caption{$\decap\,(\cdot)$}
 \label{alg:Decap}
 \end{algorithm}

\subsection{Complexity}

\subsubsection{Asymptotically Fastest Methods}

It is essential for a cryptosystem that key generation, encryption, and decryption can be implemented fast.
The following results were not known when the original FL system was proposed, but have a major impact on its efficiency.

The complexity of key generation and encryption is dominated by the cost of encoding a Gabidulin code (Line~\ref{line:encrypt_keygen} of Algorithm~\ref{alg:modKeyGen} and Line~\ref{line:encrypt_encoding} of Algorithm~\ref{alg:Encrypt}).\footnote{Note that since $\x$ and $\z$ have coefficients in the large field $\Fqmu$, this line can be realized as encoding $u$ messages over $\Fqm$ with the generator matrix $\GG \in \Fqm^{k \times n}$ and corrupting these codewords with an error (see also Section~\ref{ssec:kpub_corrupted_interleaved_codeword} below).} The asymptotically fastest-known algorithms \cite{PuchingerWachterzeh-ISIT2016,caruso2017fast,puchinger2018fast} for this require
\begin{itemize}
\item $O^\sim(n^{\min\{\frac{\omega+1}{2},1.635\}})$ operations in $\Fqm$ {or $O^\sim(n^{\omega-2}m^2)$ operations in $\Fq$ in general\footnote{Which of the two algorithms is fastest depends on the relation between $n$ and $m$, as well as the used working basis of $\Fqm$ over $\Fq$.}} and
\item $O^\sim(n)$ operations in $\Fqm$ if the entries of $\g$ are a normal basis of $\Fqm/\Fq$,
\end{itemize}
where $\omega$ is the matrix multiplication exponent and $O^\sim$ means that $\log$ factors are neglected.

The bottleneck of decryption is (error-erasure) decoding of a Gabidulin code (Line~\ref{line:decrypt_decoding} of Algorithm~\ref{alg:Decrypt}, see also Appendix~\ref{app:decryption_as_error-erasure_decoding} below), where the asymptotically fastest algorithm costs
\begin{equation*}
O^\sim\!\left(n^{\min\{\frac{\omega+1}{2},1.635\}}\right)
\end{equation*}
operations in $\Fqm$ \cite{PuchingerWachterzeh-ISIT2016,puchinger2018fast}  or
\begin{equation*}
O^\sim(n^{\omega-2}m^2)
\end{equation*}
operations in $\Fq$ (decoder in \cite{PuchingerWachterzeh-ISIT2016} with linearized-polynomial operations in \cite{caruso2017fast}).

For small lengths $n$, the algorithms from \cite{RichterPlass_DecodingRankCodes_2004,gadouleau2008complexity,SilvaKschischang-FastEncodingDecodingGabidulin-2009,Wachterzeh_DecodingBlockConvolutionalRankMetric_2013}, which have quadratic complexity over $\Fqm$ (or cubic complexity over $\Fq$), might be faster than the mentioned algorithms due to smaller hidden constants in the $O$-notation.

\subsubsection{Timing Attacks}
{In some scenarios, resistance against timing attacks is required. Due to the fact that Step~\ref{line:encrypt_encoding} of Algorithm~\ref{alg:Encrypt} can be easily implemented in constant time, the proposed encryption algorithm does not reveal any information about secret knowledge through timing attacks. The same holds for the presented decryption algorithm since there exists an efficient constant-time decoding algorithm for Gabidulin codes~\cite{bettaieb2019preventing} and all other steps of Algorithm~\ref{alg:Decrypt} can be realized in constant time as well. 
}

\section{Difficult Problems \& Semantic Security of \SystemName}\label{sec:semantic}
In this section, we introduce problems in the rank metric that are considered to be difficult. Furthermore, we prove that the public-key encryption version of {\SystemName} is IND-CPA secure and the KEM version is IND-CCA2 secure under the assumption that there does not exist probabilistic polynomial-time algorithms that can solve them.
A detailed complexity analysis of {existing and new} algorithms solving the stated problems is given in Section~\ref{sec:sec_analysis}.

\subsection{Difficult Problems in the Rank Metric}\label{sec:problems}
{\SystemName} is based on several difficult problems which are stated in this section.
Note that the search variants of the problems correspond exactly to retrieving information about the private key from the public key (not necessarily a valid private key as explained in the following) or the plaintext from the ciphertext. The decisional problems are equivalent to distinguishing the public key or the ciphertext from random vectors.

\begin{definition}[\textsf{ResIG-Distribution}: Restricted Interleaved Gabidulin Code Distribution]
	\label{def:ResIG-Distr}
	
	\noindent Input: {$q,m,n,k,w >\lfloor\frac{n-k}{2}\rfloor ,\dimZ < \frac{w}{n-k-w}, u<w$.} \\
	Choose uniformly at random
	\begin{itemize}
		\item $\G \sample \mathcal{G}$, where $\mathcal{G}$ is the set of all generator matrices of $[n,k]$ Gabidulin codes over $\Fqm$
		\item $\M \sample \{\X \in \Fqm^{u\times k} : \rank_{q^m}(\X_{[k-u,k]}) = u\} $.
		\item {$\mathcal{A} \sample \{ \text{subspace } \mathcal{U} \subseteq \Fqm^w \, : \, \dim \mathcal{U} = \zeta, \, \mathcal{U} \text{ has a basis of full-$\Fq$-rank elements} \}$}
		\item {$\E' \sample \left\{ \begin{pmatrix}
		\s_1' \\
		\vdots \\
		\s_u'
		\end{pmatrix} \in \Fqm^{u \times w}\, : \, \langle\s_1',\dots,\s_u'\rangle_{\Fqm} = \mathcal{A}, \, \rank_{q}(\s_i') = w \, \forall \, i \right\}$}
		\item {$\vec{Q} \sample \{ \A \in \Fq^{w \times n} : \rank_{q}(\A) = w  \}$}
		\item {$\E \gets \E' \vec{Q}$}
	\end{itemize}
	Output: $(\G,\M\G+\E)$.
\end{definition}

\begin{problem}[\textsf{ResIG-Search}: Restricted interleaved Gabidulin Code Search Problem]\label{pro:restr-int-search-rsd}
	\noindent Input: $(\G,\Y)$ from \textup{\textsf{ResIG-Distribution}} with input $q,m,n,k,w,\dimZ,u$ (Definition~\ref{def:ResIG-Distr}).\\
	Goal: Find $\M\in\Fqm^{u\times k}$ and $\E \in \{\X \in \Fqm^{u \times n} : \rank_{\Fq}(\X) \leq w\}$ s.t. $\M\G+\E = \Y$.
\end{problem}
Problem~\ref{pro:restr-int-search-rsd} (\textsf{ResIG-Search}) is equivalent to decoding a codeword of a $u$-interleaved Gabidulin code that is corrupted by an error $\E$, see also Section~\ref{ssec:attack_interleaved_decoding} and is therefore the underlying problem of the structural attacks from Section~\ref{sec:eff_attack}.

Note however that not necessarily every solution of this problem can be used directly as a valid private key since some additional structure on $\E$ is introduced in {\SystemName} (i.e., Problem~\ref{pro:restr-int-search-rsd} is easier to solve than retrieving a valid private key of {\SystemName}).

\begin{problem}[\textsf{ResIG-Dec}: Restricted Interleaved Gabidulin Code Decisional Problem]\label{pro:restr-int-decision-rsd}
	
	\noindent Input: $(\G,\Y) \in \Fqm^{k\times n} \times \Fqm^{u\times n}$.\\
	Goal: Decide with non-negligible advantage whether $\Y$ came from \textup{\textsf{ResIG-Distribution}} with input $q,m,n,k,w,\dimZ,u$ (Definition~\ref{def:ResIG-Distr}) or the uniform distribution over $\Fqm^{u\times n}$.
\end{problem}

To solve \textsf{ResIG-Dec} (Problem~\ref{pro:restr-int-decision-rsd}), we do not know a better approach than trying to solve the associated \emph{search} problem (i.e., \textsf{ResIG-Search}), which is usually done for all decoding-based problems.

\begin{definition}[\textsf{ResErr-Distribution}: Restricted Error Distribution]\label{def:res-err}\\
	\noindent Input: $q,m,n,k,w,\tpub,u,\ve{\gamma}, (\G,\K) $ from \textup{\textsf{ResIG-Distribution}} (Definition~\ref{def:ResIG-Distr}).\\
	Choose uniformly at random
	\begin{itemize}
		\item $\e \sample \{ \x \in \Fqm^{n}\, : \, \rank_q(\x) = \tpub  \} $
		\item $\alpha \sample \Fqmu$
		\item $\k \gets \extsmallfield_{\vec{\gamma}}^{-1}(\K)$
		\item $\y \gets \Trmum(\alpha \k) +\vec{e} = \Trmum(\alpha \m) \G +  \Trmum(\alpha \z) + \e$
	\end{itemize}
	Output: $\y$.
\end{definition}
\begin{problem}[\textsf{ResG-Search}: Restricted Gabidulin Code Search Problem]\label{pro:fl-search-rsd}	
	
	\noindent Input: $q,m,n,k,w,\tpub,u,\ve{\gamma}, (\G,\K) $ from \textup{\textsf{ResIG-Distribution}} (Definition~\ref{def:ResIG-Distr}), $\y$ from \textup{\textsf{ResErr-Distribution}} (Definition~\ref{def:res-err}) with input $(\G,\K)$.\\
	Goal: Find $\m\in\Fqm^{k}$ and $\e \in \{\x \in \Fqm^{n} : \rank_{\Fq}(\x) \leq \tpub\}$ such that $\m\G+\e = \y$.
\end{problem}
Problem~\ref{pro:fl-search-rsd} is equivalent to decoding a codeword of a Gabidulin code that is corrupted by an error that has with high probability a rank weight of $> (n-k)/2$, see Appendix~\ref{app:probability_of_large_error_weight}.

\begin{problem}[\textsf{ResG-Dec}: Restricted Gabidulin Code Decisional Problem]\label{pro:fl-decision}\\	
	\noindent Input: $q,m,n,k,w,\tpub,u,\ve{\gamma}, (\G,\K) $ from \textup{\textsf{ResIG-Distribution}} (Definition~\ref{def:ResIG-Distr}), $\y \in \Fqm^{n}$.\\
	Goal: Decide with non-negligible advantage whether $\y$ came from \textup{\textsf{ResErr-Distribution}} with input $q,m,n,k,w,\tpub,u,\ve{\gamma}, (\G,\K)$  or the uniform distribution over $ \Fqm^{n}$.
\end{problem}
As before, we are not aware of a faster approach to solve \textsf{ResG-Dec} than through the solution of the associated \emph{search} problem.

{
We will see in the next subsection that {\SystemName} is IND-CCA2 secure under the assumption that $\textsf{ResG-Dec}$ is a hard problem. As mentioned above, there is an obvious reduction of $\textsf{ResG-Dec}$ to  $\textsf{ResG-Search}$, which can again be efficiently reduced to $\textsf{ResIG-Search}$.
In fact, all relevant attacks studied in Section~\ref{sec:sec_analysis} make use of this chain of reduction and aim at solving one of the two search problems.

We are not aware of a reduction of $\textsf{ResIG-Dec}$ to $\textsf{ResIG-Search}$ or one of the other problems.
Hence, it might very well be that $\textsf{ResIG-Dec}$ is significantly easier than the other problems.
In Section~\ref{ssec:sec_analysis_IG-Dec}, we show that there is a distinguisher for $\textsf{ResIG-Dec}$ that is efficiently computable if the system parameter $\zeta$ is chosen too small.
Due to the missing reduction, it is not clear whether or not this distinguisher influences the security of the system.}

\subsection{Semantic Security}\label{ssec:semantic_security}
In this section, we prove that the public key encryption system $\sys$ is semantically secure against chosen plaintext attacks in the standard model under the assumption that \textsf{ResG-Dec} (Problem~\ref{pro:fl-decision}) is difficult. In addition, we show that the IND-CCA2 security of $\syskem$ reduces tightly to the IND-CPA security of $\sys$ in the random oracle model.

\subsubsection{IND-CPA Security of $\sys$}

To show that $\sys$ is secure against chosen plaintext attacks, we use the definition of admissibility as in~\cite{nojima2008semantic}.
\begin{definition}[Admissibility~\cite{nojima2008semantic}]\label{def:admissible}
The public key encryption scheme $\sys = (\keyGen,\enc,\dec)$ with a message space $\mathcal{M}$ and a random space $\mathcal{R}$ is called admissible if there is a pair of deterministic polynomial-time algorithms $\enc_1$ and $\enc_2$ satisfying the following property:
\begin{itemize}
\item Partible: $\enc_1$ takes as input a public key $\pk$ and $r \in \mathcal{R}$, and outputs a $p(\secLevel)$ bit-string, where $\secLevel$ is the security parameter. $\enc2$ takes as input a key $\pk$, and $\m \in \mathcal{M}$ and outputs a $p(\secLevel)$ bit-string. Here $p$ is some polynomial in the security parameter $\secLevel$. Then for any $\pk$ given by $\keyGen$, $r \in \mathcal{R}$, and $\m \in \mathcal{M}$, $\enc_1( \pk, r ) \oplus  \enc_2 ( \pk, \m) = \enc ( \pk, \m; r ) $.
\item Pseudorandomness: Let $D$ be a probabilistic algorithm and let
  \begin{align*}
    \Adv_{D,\enc_1}^{ind}(\secLevel) =& \Pr \Big[ D(\pk,\enc_1(\pk,r)) = 1 \mid r \sample \mathcal{R}, (\sk, \pk) \gets \keyGen\big(1^{\secLevel}\big) \Big] \\
    & - \Pr \big[ D(\pk,s) = 1 \mid s \sample \mathcal{U}_{p(\secLevel)}, (\sk, \pk) \gets \keyGen\big(1^{\secLevel} \big)    \big].
  \end{align*}
  We define the advantage function of the problem as follows. For any $t$,
  \begin{equation*}
\Adv_{\enc_1}^{ind}(\secLevel,t) = \max_{D} \big\{\Adv_{D,\enc_1}^{ind}(\secLevel) \big\},
\end{equation*}
where the maximum is taken over all $D$ with time-complexity $t$. Then, the function $\Adv_{\enc_1}^{ind}(\secLevel,t)$ is negligible for every polynomial bounded $t$ and every sufficiently large $\secLevel$.
\end{itemize}
\end{definition}

In the following we will prove that $\sys$ is IND-CPA secure by showing that is fulfills the definition of admissibility.
\begin{theorem}\label{thm:cpa}
The  system $\sys = (\keyGen,\enc, \dec)$ is an IND-CPA secure encryption scheme in the standard model under the assumption that the \textup{\textsf{ResG-Dec}} problem is difficult.
\end{theorem}
\begin{proof}
  Let
  $\enc_1 := \Trmum (\alpha \kpub) +\e$
and
$\enc_2 := \m \GG$.
Then, one observes that $\enc = \enc_1 \oplus \enc_2$ and thus $\sys$ is \emph{partible}.
Since \textup{\textsf{ResG-Dec}} (Problem~\ref{pro:fl-decision}) is assumed to be difficult, the encryption scheme fulfills \emph{pseudorandomness} and thus, the system is \emph{admissibile}. As proven in~\cite[Lemma 1]{nojima2008semantic}, if $\sys$ fulfills Definition~\ref{def:admissible}, then it is an IND-CPA secure encryption scheme.
\qed
\end{proof}

\subsubsection{IND-CCA2 Security of $\syskem$}

We used a transformation proposed in~\cite{hofheinz2017modular} to transform the public key encryption scheme $\sys$ into the KEM $\syskem$. In the following, we prove that $\syskem$ is IND-CCA2 secure.

The applied transformation requires that the encryption scheme is $\gamma$-spread which is proven to be the case for $\sys$ in the following.
\begin{definition}[$\gamma$-spread, \cite{fujisaki12013secure,hofheinz2017modular}] For valid $(\pk,\sk)$, the \emph{min-entropy} of\\ $\enc(\m,\pk)$ is defined by 
  \begin{equation*}
    \gamma(\m,\pk) := - \log_2 \max_{\c \in \cipherset} \Pr_{r\gets \mathcal{R}} [\c = \enc(\m,\pk;r)],
  \end{equation*}
  where $\cipherset$ is the set of possible ciphertexts. A public key encryption scheme is called \emph{$\gamma$-spread} if for every valid key pair $(\pk,\sk)$ and every message $\m\in \mathcal{M}$, $\gamma(\m,\pk) \geq \gamma$. It follows that for all $\c \in \cipherset$,
  \begin{equation*}
 \Pr_{r\gets \mathcal{R}} [\c = \enc(\m,\pk;r)] \leq 2^{- \gamma}.
  \end{equation*}
\end{definition}

\begin{lemma}\label{lem:gamma-spread}
The public key encryption system $\sys$ is $\gamma$-spread, where $\gamma = m(\tpub-u)+\tpub(n-\tpub-1)$.
\end{lemma}
\begin{proof}
  We observe that
  \begin{align*}
     \max_{\c \in \cipherset}\Pr_{r\gets \mathcal{R}} [\c = \enc(\m,\pk;r)] 
    {=}& \max_{\c \in \cipherset} \Pr_{r\gets \mathcal{R}} [\c = (\vec{m},\vec{0}_u)\mathbf G_{\mycode{G}} + \Trmum(\alpha \kpub) +\vec{e}] \\
        \overset{\text{(i)}}{\leq}& \max_{\c' \in \cipherset'}q^{mu}\Pr_{r\gets \mathcal{R}} [\c' = (\vec{m},\vec{0}_u)\mathbf G_{\mycode{G}} +\vec{e}] \\
        {=}&~ q^{mu} \frac{1}{|\{\e \in \Fqm^{n} : \rank_{q}(\e) = \tpub \}|},
\end{align*}
where $\cipherset'$ is the set of all vectors in rank distance $\tpub$ from $(\vec{m},\vec{0}_u)\mathbf G_{\mycode{G}}$ and $\text{(i)}$ follows from the fact that there at most $q^{mu}$ choices for  $\alpha$.
In \cite[Section~IV.B]{gabidulin2003reducible}, a constructive way of obtaining rank-$\tpub$ matrices is given. More precisely, an injective mapping $ \varphi \, : \, \Fq^{t(n+m-t-1)} \to \{\A \in \Fq^{n \times m} \, : \, \rank \A = t\}$ is given. Hence, we have  
$|\{\e \in \Fqm^{n} : \rank_{q}(\e) = \tpub \}| \geq q^{t(n+m-t-1)}$. It follows that
\begin{align*}
  \frac{q^{mu}}{ |\{\e \in \Fqm^{n} : \rank_{q}(\e) = \tpub \}|}
  &\leq \frac{q^{mu}}{ q^{\tpub(n+m-\tpub-1)} } \\
  &= q^{-m(\tpub-u)-\tpub(n-\tpub-1)} \\
  &\leq q^{-\gamma}. 
\end{align*}
\qed
\end{proof}

\begin{theorem}\label{thm:cca2}
 The KEM-DEM $\syskem = (\keyGen,\encap, \decap)$ is IND-CCA2 secure in the random oracle model under the assumption that the \textup{\textsf{ResG-Dec}} problem is difficult.
\end{theorem}
\begin{proof}
Assuming the \textsf{ResG-Dec} is difficult, the encryption $\sys$ is IND-CPA secure, see Theorem~\ref{thm:cpa}. Further, it is proven in Lemma~\ref{lem:gamma-spread} that $\sys$ has $\gamma$-spread encryptions. Thus, the system $\syskem$ can be tightly reduced to $\syskem$ in the random oracle model as shown in~\cite{hofheinz2017modular}. \qed 
\end{proof}

\section{Security Analysis of {\SystemName}}\label{sec:sec_analysis}

In this section, we analyze the security of \SystemName. As proven in Theorem~\ref{thm:cpa} and~\ref{thm:cca2}, the encryption version is IND-CPA secure and the KEM version is IND-CCA2 secure under the assumption that \textsf{ResG-Dec} is difficult.
Since there are obvious reductions from \textsf{ResG-Dec} to \textsf{ResG-Search} and from \textsf{ResG-Dec} to \textsf{ResIG-Search}, we will study the hardness of these two search problems in this section (Section~\ref{ssec:attacks_public_key} for \textsf{ResIG-Search} and Section~\ref{ssec:attacks_ciphertext} for \textsf{ResG-Search}).
In fact, we are not aware of a more efficient method to solve \textsf{ResG-Dec} than through these two search problems.

Although no formal reduction from any of the other three studied problems to \textup{\textsf{ResIG-Dec}} is known, we study also the hardness of \textup{\textsf{ResIG-Dec}} (Section~\ref{ssec:sec_analysis_IG-Dec}).
We derive a distinguisher for the public key with exponential complexity in the system parameters, which can be avoided by proper parameter choice.

Due to the nature of the (random) encryption, there are public keys for which the probability that the work factor of some of the ciphertext attacks on \textsf{ResG-Search} (ciphertext attack) is below the designed minimal work factor is not negligible (i.e., $> 2^{-\secLevel}$). We show in Section~\ref{ssec:weak_keys} that these weak keys occur with negligible probability (i.e., $\leq 2^{-\secLevel}$) during the random key generation if the parameters are chosen in a suitable way.

\subsection{Exponential-Time Attacks on \textup{\textsf{ResIG-Search}}}\label{ssec:attacks_public_key}
We propose new and summarize known methods that solve \textup{\textsf{ResIG-Search}} (Problem~\ref{pro:restr-int-search-rsd}).
All studied algorithms have exponential complexity in the code parameters.

Recall that in the decryption algorithm of {\SystemName}, the last $u$ positions of the private key $\x$ have to be a basis of $\Fqmu$ over $\Fqm$. Therefore, not every solution of Problem~\ref{pro:restr-int-search-rsd} (\textsf{ResIG-Search}) can be used as valid private key and Problem~\ref{pro:restr-int-search-rsd} is a strictly easier problem than retrieving a valid private key corresponding to a given public key.

\subsubsection{Brute-Force the Vector $\z$ Attack}\label{ssec:attack_brute-force_z}
The number of vectors $\z \in \Fqmu^n$ that fulfill the conditions stated in Section~\ref{sec:repair} is equal to number of possible vectors $\s \in \Fqmu^w$ times the number of full rank matrices in $\Fqm^{w\times n}$ in reduced row Echelon form. Formally, the number of vectors $\z$ is
\begin{align*}
  \underbrace{\left|\left\{ \z \, : \, \z \text{ can occur in Alg.~\ref{alg:modKeyGen}}\right\}\right|}_{\geq 1} \cdot \underbrace{\left|\left\{ \P \, : \, \P \text{ can occur in Alg.~\ref{alg:modKeyGen}}\right\}\right|}_{\geq \qbinomialField{n}{w}{q}} \geq \qbinomialField{n}{w}{q}.    
\end{align*}
Thus, brute-forcing a vector $\z$ that is a solution to \textsf{ResIG-Search} has work factor
\begin{equation*}
  \text{WF}_{\text{\vec{z}}}\geq\frac{\qbinomialField{n}{w}{q}}{\mathcal{N}} \geq \frac{q^{w(n-w)}}{\mathcal{N}}, 
\end{equation*}
where the latter inequality follows from a lower bound on  $q$-binomials (see~\cite[Lemma 4]{koetter2008coding}), and 
\begin{equation}\label{eq:avNumberIntCodewords}
\mathcal{N} = \max \left\{|\mathcal{C}| \cdot 
 \frac{\sum_{i=0}^{w-1}{\prod_{j=0}^{i-1}{(q^{mu} - q^j)}  {\displaystyle\quadbinom{n}{i}}} }{q^{mun}}
 , \
 1\right\}
\end{equation}
 is the average number of interleaved codewords in a ball of radius $w$ around a uniformly at random chosen interleaved received word.

\subsubsection{Interleaved Decoding Attack}\label{ssec:attack_interleaved_decoding}
As described in Section~\ref{sec:eff_attack}, an attacker can apply an interleaved decoder on $\kpub$ to retrieve an alternative private key. A major ingredient of {\SystemName} is that the public key is chosen in a way that this decoding will always fail (i.e., the corresponding linear system of equations does not have a unique solution). However, it is still possible to brute-force search in the solution space of the involved system of equations. This is analyzed in the following. Notice thereby that \emph{any} interleaved codeword in radius at most $w$ is a solution to \textsf{ResIG-Search}.

Problem~\ref{pro:restr-int-search-rsd} (\textsf{ResIG-Search}) is equivalent to decoding a codeword of a $u$-interleaved Gabidulin code that is corrupted by an error $\E$. The error $\E$ fulfills 
\begin{equation*}
\bigg\lfloor \frac{n-k}{2} \bigg\rfloor < \rank_{q}(\E) \leq \frac{u}{u+1} (n-k) \quad \text{and} \quad \rank_{q^m}(\E) < \frac{w}{n-k-w} < w 
\end{equation*}
and thus, no known algorithm is able to correct it efficiently.

The crucial point of the interleaved decoding algorithms from \cite{Loidreau_Overbeck_Interleaved_2006,Sidorenko2011SkewFeedback} is solving a linear system of equations based on the syndromes with $w+1$ unknowns and $\RkError$ linearly independent equations which is equivalent to finding the kernel of the matrix in~\eqref{eq:matrix-repair}, cf.~\cite[Section~4.1]{Wachterzeh_DecodingBlockConvolutionalRankMetric_2013}. For $\dimZ \geq \frac{w}{n-k-w}$, the dimension of the solution space is one and all solutions are valid for the remaining decoding steps. For $\dimZ < \frac{w}{n-k-w}$, the dimension of the solution space is $w+1-\RkError$ but each valid solution forms only a one-dimensional subspace. An attacker can therefore search in the solution space for a valid solution which requires on average
\begin{equation*}
  \frac{(q^m)^{w+1-\RkError}}{q^m \cdot \mathcal{N}} = \frac{q^{m(w-\RkError)}}{\mathcal{N}}
\end{equation*}
 trials, where $\mathcal{N}$ is the average number of interleaved codewords, see~\eqref{eq:avNumberIntCodewords}.

The size of the solution space is $w+1-\RkError$  and clearly maximized for the smallest-possible value of $\RkError$, i.e., $\RkError = n-k-w$.
In this case, the search through the solution space has work factor
 \begin{equation*}
\text{WF}_{\text{ILD}}=\frac{q^{m(w-\dimZ(n-k-w))}}{\mathcal{N}}.
\end{equation*}
Since the size of the solution space is maximal for $\RkError = n-k-w$, the repair from Section~\ref{sec:repair} with the explicit parameter value $\dimZ=1$ (i.e., $\dim \big( \langle \vec{z}_1 , \dots , \vec{z}_u\rangle_{\Fqm}\big) = 1$) is the most secure choice \emph{in this sense}. 
However, we keep the choice of $\dimZ$ flexible as the pair-wise linear dependence of the $\z_i$ might decrease the security (we are however not aware of how this fact could be used).

Besides the syndrome-based interleaved decoding algorithms in \cite{Loidreau_Overbeck_Interleaved_2006}, \cite{Sidorenko2011SkewFeedback}, and \cite[p.~64]{Wachterzeh_DecodingBlockConvolutionalRankMetric_2013}, there is an interpolation-based decoding algorithm \cite[Section~4.3 (page 72)]{Wachterzeh_DecodingBlockConvolutionalRankMetric_2013}.
This interpolation-based algorithm can be interpreted both as a list decoder of interleaved Gabidulin codes with exponential worst-case and average list size or as a probabilistic unique decoder.
The probabilistic unique interpolation-based decoder fails if and only if the decoding algorithms in \cite{Loidreau_Overbeck_Interleaved_2006}, \cite{Sidorenko2011SkewFeedback}, \cite[p.~64]{Wachterzeh_DecodingBlockConvolutionalRankMetric_2013} fail and therefore the previous analysis applies here as well.
For the list decoder, cf. \cite[Lemma~4.5]{Wachterzeh_DecodingBlockConvolutionalRankMetric_2013}, the work factor of the resulting attack is 
\begin{equation}\label{eq:WF_list_size_kpub}
\text{WF}_\text{list, public key} \leq \frac{q^{m(u-1)k}}{\mathcal{N}}.
\end{equation}
Notice that the list of size $q^{m(u-1)k}$ contains many words which are no valid codewords, but we have to go through the whole list to find all valid codewords in radius up to $w$.

\subsubsection{List Decoding of the Public Key Attack}\label{ssec:attack_list_pk}
Recall that  $\kpub = \vec{x} \cdot \mathbf G_{\mycode{G}} + \vec{z}$. 
Previously, we have explained why this vector is a corrupted version of a codeword of a $u$-interleaved Gabidulin code.
At the same time, $\vec{x} \cdot \mathbf G_{\mycode{G}}$ can be seen as a short Gabidulin code over a large field $\Fqmu$ and therefore, if existing, one could apply list decoding algorithms to decode $\kpub$ and obtain $\vec{x}$. The weight of the error $\vec{z}$ is larger than the unique decoding radius and therefore a unique decoder cannot be applied to reconstruct $\vec{x}$ and a list decoder for radius $w$ is required.

However, such an algorithm has not been found yet. It was even shown in \cite{Wachterzeh_BoundsListDecodingRankMetric_IEEE-IT_2013,RavivWachterzeh_GabidulinBounds_journal,ListDec_RankMetric_2019} that for most classes of Gabidulin codes such a polynomial-time list decoding algorithm cannot exist. 
Note that these results were not known when the original FL cryptosystem was proposed.
These results also imply that there is no polynomial-time list decoding algorithm for arbitrary Gabidulin codes beyond the unique decoding radius (such as the Guruswami--Sudan algorithm for Reed--Solomon codes).

\subsubsection{Randomized Gabidulin Decoding Attack on the Public Key}\label{ssec:attack_rsd_pk}

The public key can be seen as the sum of a Gabidulin codeword over the field $\Fqmu$ and an error of weight $w > \frac{n-k}{2}$. Alternatively, as shown in Section~\ref{ssec:kpub_corrupted_interleaved_codeword}, the public key can be seen as an interleaved Gabidulin codeword that is corrupted by an error of weight $w$ ({note that this is the reason why all the $\s_i$'s must have full $\Fq$-rank in Algorithm~\ref{alg:modKeyGen}}). Each row of \eqref{eq:kpub_as_corrupted_interleaved_codeword} is a codeword of a Gabidulin code over $\Fqm$ that is corrupted by an error of rank weight $w$. Both the corrupted Gabidulin codeword over $\Fqmu$ as well as over $\Fqm$ can be decoded using the randomized decoding approach proposed in~\cite{renner2019randomized}.
Since applying the attack on each row of the unfolded public key is more efficient, we conclude that the randomized Gabidulin decoding attack on the public key has an average complexity of  
\begin{equation*}
\text{WF}_{\text{RGD}} = \frac{n}{64}\cdot q^{m(n-k)-w(n+m)+w^2+\min\{2\xi(\frac{n+k}{2}-\xi),wk\} }
\end{equation*}
over $\Fqm$.

\subsubsection{Moving to Another Close Error Attack}\label{ssec:attack_moving_to_error}
The following attack was suggested by Rosenkilde~\cite{Johan_attack}.
It tries to move the vector $\vec{z}$ (which we have chosen such that the interleaved decoder fails) to a close vector of the same or smaller rank weight $w$ for which the interleaved decoder for $\kpub$ does not fail.

The idea is to find a vector $\vec{y} \in \Fqm^{u \times n}$ such that $\vec{z}^\prime := \vec{z} + \vec{y}$ still has rank weight $\rank_q(\vec{z}^\prime) \leq w$ and that the rank of the matrix from~\eqref{eq:matrix-repair} over $\Fqm$ is at least $w$.
To guarantee the first condition, we want to construct $\vec{y}$ such that its extended $um \times n$ matrix over $\Fq$ has a row space $\mathcal{R} := \mathrm{RowSpace}_{\Fq}\big(\extsmallfield_{\boldsymbol{\gamma}}(\y)\big)$ that is contained in the one of $\vec{z}$.
Since for the original error $\z$, the matrix \eqref{eq:matrix-repair} has rank $\varphi \leq \dimZ(n-k-w)$, $\mathcal{R}$ must have at least $\Fq$-dimension $w-\varphi \geq w(\dimZ+1)-\dimZ(n-k)$.
By choosing a random $\mathcal{R}$ with this property and taking a random matrix $\vec{y}$ whose extended matrix has $\Fq$-row space $\mathrm{RowSpace}_{\Fq}\big(\extsmallfield_{\boldsymbol{\gamma}}(\y)\big) = \mathcal{R}$, the second condition is fulfilled with high probability.

The complexity of the attack is hence dominated by the complexity of finding a subspace $\mathcal{R} \subseteq \Fq^{n}$ of dimension $w-\varphi$ that is contained in the $w$-dimensional $\Fq$-row space of $\z$.
Since this is unknown, we can find it in a Las-Vegas fashion by repeatedly drawing a subspace uniformly at random.
The expected number of iterations until we find a suitable row space is thus one over the probability that a random $(w-\varphi)$-dimensional subspace of $\Fq^n$ is contained in a given $w$-dimensional subspace, which is (cf.~\cite[Proof of Lemma~7]{etzion2011error}):
\begin{align*}
\frac{\qbinomial{w}{w-\varphi}}{\qbinomial{n}{w-\varphi}} \approx \frac{q^{\varphi(w-\varphi)}}{q^{(n-w+\varphi)(w-\varphi)}} = q^{-(n-w)(w-\varphi)} \leq q^{-(n-w)(w(\dimZ+1)-\dimZ(n-k))}.
\end{align*}
Hence, the attack has work factor
\begin{align*}
\text{WF}_{\text{MCE}} &= q^{(n-w)(w-\varphi)} \geq q^{(n-w)(w(\dimZ+1)-\dimZ(n-k))}.
\end{align*}

\subsection{Exponential-Time Attacks on \textup{\textsf{ResG-Search}}}\label{ssec:attacks_ciphertext} 
Retrieving information about the plaintext from the ciphertext and the public key is equal to solving \textup{\textsf{ResG-Search}} (Problem~\ref{pro:fl-decision}). In this section, methods that solve this problem are summarized.

\subsubsection{Randomized Gabidulin Decoding Attack on the Ciphertext}\label{ssec:attack_rsd_cipher}

Each ciphertext of {\SystemName} can be seen as a Gabidulin codeword over $\Fqm$ plus an error:
\begin{align*}
\vec{c} &= (\vec{m},\vec{0}_u) \cdot \mathbf G_{\mycode{G}} + \Trmum(\alpha \kpub) + \vec{e} \\
&= \underbrace{\left[(\vec{m},\vec{0}_u) + \Trmum(\alpha \vec{x})\right] \cdot \mathbf G_{\mycode{G}}}_{\text{codeword}} + \underbrace{\Trmum(\alpha\vec{z}) + \vec{e}}_{\text{error}}
\end{align*}
Denote $\tilde{w} := \rank_{\Fq}(\Trmum(\alpha\vec{z}) + \vec{e})$.
Then we can use the decoding algorithm proposed in~\cite{renner2019randomized}, which requires on average at least
\begin{equation}
\frac{n}{64}\cdot q^{m(n-k)-\tilde{w}(n+m)+\tilde{w}^2+\min\{2\xi(\frac{n+k}{2}-\xi),\tilde{w}k\} } \label{eq:WF_RGD}
\end{equation}
operations in $\Fqm$.

Clearly, the complexity of the algorithm strongly depends on the value $\tilde{w}$, which in turn depends on the generated keys.
{In general, $\tilde{w} = w +\tpub$, but for some choices of $\z$, $\alpha$, and $\e$, the rank $\tilde{w}$ is smaller.
For this issue, see Section~\ref{ssec:weak_keys} and Appendix~\ref{app:probability_of_large_error_weight}, where we study the probability that $\tilde{w}$ is small, both for randomness in the encryption (random choice of $\alpha$ and $\e$) and the key generation (random choice of $\z$).
Some extremely rarely occurring keys (weak keys) thereby result in relatively high probabilities that $\tilde{w}$ is small.

However, we can choose the system parameters such that both, the probability of a weak key as well as the conditional probability that $\tilde{w}<w$ given a non-weak key is below $2^{-\secLevel}$.
Hence, with overwhelming probability, a random key and ciphertext result in a ciphertext error of rank weight $\tilde{w}\geq w$ and the work factor of this attack is always at least as large as the ``Randomized Gabidulin Decoding Attack on the Public Key'' in Section~\ref{ssec:attacks_public_key}.}

\subsubsection{List Decoding of the Ciphertext Attack}\label{ssec:list_decoding_attack_ciphertext}
As described in the Randomized Gabidulin Decoding Attack on the Ciphertext above, the ciphertext of {\SystemName} is a codeword of a Gabidulin code, corrupted by an error of rank weight $\tilde{w}$,  Hence, an attacker can try to decode the ciphertext directly.
Since $\tau$ is always greater than the unique decoding radius $\nkhalffrac$ of the Gabidulin code, this would require the existence of an efficient (list) decoding algorithm up to radius $\tau$.
As explained previously, there is no such algorithm and bounds on the list size prove that there cannot exist a generic list decoding algorithm for all Gabidulin codes which indicates that list decoding is a hard problem.

However, to be secure, we have considered list decoding as follows for the security level of our system.
The list size $\mathcal{L}_\text{$\c$,worst}$ denotes a lower bound on the \emph{worst-case} work factor of list decoding.
For example, for a Gabidulin code with parameters $n\mid m$ and $\gcd(n,n-\tau)\geq 2$, there is a received word such that there are at least
\begin{align}\label{eq:list_size_cipher}
\mathcal{L}_\text{$\c$,worst} \geq \max\left\{
\frac{\qbinomialField{n/g}{(n-\tau)/g}{q^g}}
{q^{n (\tau/g-1)}} \, : \, g \geq 2, \, g \mid \gcd(n,n-\tau) \right\}
\end{align}
codewords in rank distance at most $\tau$ to it.

Although $\mathcal{L}_\text{$\c$,worst}$ does not imply any statement about the average list size/average work factor, it provides an estimate of the order of magnitude of the work factor of a hypothetical list decoding attack.
For our suggested parameters, we have ensured that the value of $\mathcal{L}_\text{$\c$,worst}$ 
is sufficiently large in the proposed sets of parameters in Section~\ref{sec:parameters}.

\subsubsection{Combinatorial Rank Syndrome Decoding (RSD) Attack}\label{ssec:attack_syndrome_decoding}
The ciphertext can be interpreted as a codeword from a code of dimension $k$ (see~\cite{faure2006new}), generated by the generator matrix
\begin{equation*}
\begin{pmatrix}
\Mooremat{k-u}{q}{\vec{g}}\\
\Trmum(\gamma_1 \kpub)\\
\vdots\\
\Trmum(\gamma_u \kpub)\\
\end{pmatrix}.
\end{equation*}
Since the structure of this code only permits decoding like a random rank-metric code, it can be decoded with the combinatorial syndrome decoding attack from~\cite{Aragon2018-DecAttack} whose complexity is in the order of
\begin{align*}
\text{WF}_{\text{CRSD}}=(n-k)^3m^3q^{\tpub \lceil \frac{(k+1)m}{n} \rceil -m}.  
\end{align*}

\subsubsection{Algebraic RSD Attack}\label{ssec:attack_algebraic_syndrome_decoding}
As described in the previous section, the \textsf{ResG-Search} problem can be solved by decoding an error of rank weight $\tpub$ in a random $[n,k]$ code. Beside the combinatorial approach, there exist algebraic algorithms to solve the Problem.

In~\cite{B3GNRT19}, the RSD problem is expressed as a multivariate polynomial system and is solved by computing a Gröbner basis. The complexity of that attack is generally smaller than the combinatorial approach. In case there is a unique solution to the system, then the work factor of the algorithm is 
\begin{equation*}
  \text{WF}_{\text{Gr}} = 
\begin{cases}
  \left[ \frac{((m+n)\tpub)^{\tpub}}{\tpub!} \right]^\mu & \text{if $m{n-k-1 \choose \tpub } \le {n \choose \tpub}$} \\[8pt]
  \left[ \frac{((m+n)\tpub)^{\tpub+1}}{(\tpub+1)!} \right]^\mu & \text{otherwise, }
  \end{cases}
  \end{equation*}
where $\mu = 2.807$ is the linear algebra constant.

{Very recently, a new algebraic algorithm was proposed to solve the RSD problem~\cite{bardet2020algebraic}. It divides the RSD problem instances into two categories. If
\begin{equation*}
 m\binom{n-k-1}{\tpub} \geq \binom{n}{\tpub} -1,
\end{equation*}
we are in the overdetermined case and the proposed algorithm has work factor
\begin{equation*}
 \text{WF}_{\text{Wogr}} = m \binom{n-p-k-1}{\tpub} \binom{n-p}{\tpub}^{\mu-1}
\end{equation*}
in $\Fq$, where $p= \min \{i: i \in \{1,\hdots,n\}, m \binom{n-i-k-1}{\tpub} \geq \binom{n-i}{\tpub}-1\}$. Otherwise, we are in the underdetermined case in which the algorithm has work factor
\begin{equation*}
  \text{WF}_{\text{Wogr}} = \min \{ \text{WF}_{\text{Under}}, \text{WF}_{\text{Hybrid}} \}.
\end{equation*}
We have
\begin{equation*}
  \text{WF}_{\text{Hybrid}} = q^{a \tpub} m \binom{n-k-1}{\tpub} \binom{n-a}{\tpub} ^{\mu-1}
\end{equation*}
with $a = \min \{i: i \in \{1,\hdots,n\}, m \binom{n-k-1}{\tpub} \geq \binom{n-i}{\tpub}\ - 1\}$. Further, for $0<b<\tpub+2$ and $A_b -1 \leq B_b + C_b$,
\begin{equation*}
  \text{WF}_{\text{Under}} = \frac{B_b \binom{k+\tpub+1}{\tpub} + C_b(mk+1)(\tpub+1)}{B_b+C_b}\Bigg( \sum_{j=1}^{b} \binom{n}{\tpub} \binom{mk+1}{j} \Bigg)^2,
\end{equation*}
where $A_b:=\sum_{j=1}^{b}\binom{n}{\tpub} \binom{mk+1}{j}$, $B_b:=\sum_{j=1}^{b} m\binom{n-k-1}{\tpub}\binom{mk+1}{j}$ and
\begin{equation*}
  C_b := \sum_{j=1}^{b} \sum_{i=1}^{j} \bigg( (-1)^{i+1} \binom{n}{\tpub+i} \binom{m+i-1}{i} \binom{mk+1}{j-i} \bigg).
\end{equation*}

We denote the minimum of the work factors of the two algorithms as the work factor of the algebraic RSD attack, i.e.,
\begin{equation*}
  \text{WF}_{\text{ARSD}} = \min\{ \text{WF}_{\text{Gr}}, \text{WF}_{\text{Wogr}} \}.
\end{equation*}
Note that for algebraic decoding, it is neither known how to improve the complexity by using the fact that there are multiple solutions, nor it is known how to speed up the algorithm in the quantum world. 
}

\subsubsection{Linearization Attack}\label{ssec:attack_linearization}
In~\cite{faure2006new}, a message attack was proposed which succeeds for some parameters with high probability in polynomial time.
\begin{lemma}[Linearization Attack {\cite{faure2006new}}]
Let $\kpub^{(i)} = \Trmum(\gamma_i \kpub)$ for $i=1,\dots,u$ and
\begin{align}
\M =
\begin{pmatrix}
\Mooremat{\tpub+1}{q}{\c} \\
-\Mooremat{\tpub+1}{q}{\kpub^{(1)}} \\
\vdots \\
-\Mooremat{\tpub+1}{q}{\kpub^{(u)}} \\
-\Mooremat{k+\tpub-u}{q}{\g}
\end{pmatrix}. \label{eq:M_linearization_attack}
\end{align}
Then, the encrypted message $\m$ can be efficiently recovered if the left kernel of $\M$ has dimension $\dim(\ker(\M)) = 1$.
\end{lemma}

If $(u+2)\tpub + k > n$, then $\M$ has at least two more rows than columns and we have $\dim(\ker(\M))>1$.
If $\kpub$ is random and $(u+2)\tpub + k \leq n$, the attack is efficient with high probability~\cite{faure2006new}.

\begin{lemma}
Let $\M$ be as in \eqref{eq:M_linearization_attack}. Then,
\begin{align*}
\rank_{q^m}(\M) \leq \min\{\varphi + k + 2\tpub -u,n \}.
\end{align*}
\end{lemma}

\begin{proof}
We can write
\begin{align*}
\kpub^{(i)} &= \Trmum(\gamma_i \kpub) \\
&= \Trmum(\gamma_i \vec{x}) \cdot \Mooremat{k}{q}{\g} + \vec{z}_i,
\end{align*}
so by elementary row operations, we can transform $\M$ into
\begin{align*}
\M' = \begin{pmatrix}
\Mooremat{\tpub+1}{q}{\c} \\
-\Mooremat{\tpub+1}{q}{\vec{z}_1} \\
\vdots \\
-\Mooremat{\tpub+1}{q}{\vec{z}_u} \\
-\Mooremat{k+\tpub-u}{q}{\g}
\end{pmatrix}.
\end{align*}
Due to $w+2\tpub<n-k$, the matrix $\Mooremat{\tpub+1}{q}{\vec{z}_i}$ is a sub-matrix of $\Mooremat{n-k-w}{q}{\vec{z}_i}$, so
\begin{align*}
&\rank_{q^m} (\M) = \rank_{q^m} (\M') \\
&\leq \varphi + \rank_{q^m} (\Mooremat{\tpub+1}{q}{\c}) + \rank_{q^m} (\Mooremat{k+\tpub-u}{q}{\g}) \\
&= \varphi + k + 2\tpub -u.
\end{align*}
Further, since the number of columns of $\M$ is equal to $n$,
\begin{equation*}
\rank_{q^m} (\M) \leq n.
\end{equation*}
\qed
\end{proof}

The linearization attack is inefficient if the rank of $\M$ is smaller than its number of rows, which implies the following, stronger version of the original statement in~\cite{faure2006new}.
\begin{theorem}\label{cor:linearization_new}
  If $\tpub > \tfrac{n-k}{u+2}$ or $\varphi < u (\tpub+1)$, the linearization attack in \cite{faure2006new} is inefficient and its work factor is
  \begin{align*}
    \text{WF}_{\mathrm{Lin}} = q^{m \cdot \max\{u \tpub + u + 1 - \varphi,(u+2)\tpub + k +1 -n\}}.
  \end{align*}
\end{theorem}
The first condition in Corollary~\ref{cor:linearization_new} is again fulfilled by the choice of $w$ in Table~\ref{tab:parameters}.
The second one reads
$\tpub > \tfrac{\varphi}{u}+1$,
and for any valid $\varphi$, there are choices of $w$ such that $\tpub$ fulfills this inequality for any $u>1$.

\subsubsection{Algebraic Attacks}\label{ssec:algebraic_attacks}
Faure and Loidreau \cite{faure2006new} also described two message attacks of exponential worst-case complexity. The first one is based on computing gcds of polynomials of degrees and has a work factor
\begin{equation}
\text{WF}_{\text{GCD}}=
 q^{m(u-1)}\frac{q^{\tpub+1}-1}{q-1}. \label{eq:algebraic_attack_work_factor}
\end{equation}
Since computing the gcd of two polynomials can be implemented in quasi-linear time in the polynomials' degree,~\eqref{eq:algebraic_attack_work_factor} gives an estimate on the work factor of this attack. The second algebraic attack is based on finding Gr\"obner bases of a system of
$n_\mathrm{p} = \tbinom{n}{k+2 \tpub-u+1}$ many polynomials of degree approximately $d_\mathrm{p} = \tfrac{q^{\tpub+1}-1}{q-1}$.
The attack is only efficient for small code parameters, cf.~\cite[Sec.~5.3]{faure2006new}. Since the average-case complexity of Gr\"obner bases algorithms is hard to estimate, we cannot directly relate $n_\mathrm{p}$ and $d_\mathrm{p}$ to the attack's work factor. Faure and Loidreau choose the code parameters such that $n_\mathrm{p} \approx 2^{32}$ and $d_\mathrm{p} = 127$ and claim that the attack is inefficient for these values. Our example parameters in Section~\ref{sec:parameters} result in at least these values.

\subsubsection{Overbeck-like Attack}\label{ssec:attack_overbeck-like}
The key attack described in \cite[Ch.~7, Sec.~2.1]{LoidreauHabitilation-RankMetric_2007} is based on a similar principle as Overbeck uses to attack the McEliece cryptosystem based on Gabidulin codes~\cite{Overbeck-StructuralAttackGPT}.
The attack from \cite[Ch.~7, Sec.~2.1]{LoidreauHabitilation-RankMetric_2007} cannot be applied if
\begin{equation*}
w \geq n-k -\frac{k-u}{u-1}.
\end{equation*}

\subsubsection{Brute-Force Attack on the Element $\alpha$}\label{ssec:attack_brute-force-alpha}
An attacker can brute-force $\alpha \in \Fqmu$, which has a complexity of
\begin{align*}
\text{WF}_{\alpha}=  q^{mu}.  
\end{align*}
By knowing $\alpha$, he just needs to apply an efficient decoding algorithm on $\tilde{\vec{c}} = \vec{c} - \Trmum(\alpha \kpub)$ to retrieve the secret message.

\subsection{Exponential-Time Attacks on \textup{\textsf{ResIG-Dec}}}\label{ssec:sec_analysis_IG-Dec} 
We have seen in Section~\ref{sec:semantic} that {\SystemName} is IND-CCA2 secure under the assumption that \textup{\textsf{ResG-Dec}} is a hard problem.
The two previous subsections analyzed all known attacks on the \textup{\textsf{ResG-Search}} and \textup{\textsf{ResIG-Search}} problems, which are relevant since there is an obvious reduction of \textup{\textsf{ResG-Dec}} to these search problems.

In the following, we study Problem~\textup{\textsf{ResIG-Dec}} (which translates to distinguishing the public key from a random vector in $\Fqmu^n$), which is different in the sense that we do not know an efficient reduction from \textup{\textsf{ResG-Dec}} (or one of the search problems) to \textup{\textsf{ResIG-Dec}}.
In other words, even if distinguishing the public key is easy, it might still be hard to distinguish the ciphertext.
Nevertheless, we study the hardness of \textup{\textsf{ResIG-Dec}} in the following and present a distinguisher, which is efficient to compute if $\zeta$ is chosen small.
The distinguisher is as follows.

Recall the choice of $\kpub$ in Algorithm~\ref{alg:modKeyGen}.
We have
\begin{align*}
\kpub = \vec{x} \cdot \mathbf G_{\mycode{G}} + \vec{z} \in \Fqmu^n.
\end{align*}
Expand $\kpub$ into a $u \times n$ matrix over $\Fqm$ and choose any $\zeta+1$ rows.
As the $\Fqm$-expansion of the error $\vec{z}$ has $\Fqm$-rank $\zeta$, there are at least $q^m-1$ many non-trivial $\Fqm$-linear combinations of these $\zeta+1$ rows that are codewords of~$G_{\mycode{G}}$.
This is not true with high probability for a random $u \times n$ matrix over $\Fqm$.

Thus, by repeatedly randomly linearly combining these $\zeta+1$ rows and checking whether the result is a codeword of $G_{\mycode{G}}$, we obtain a Monte-Carlo algorithm with an expected work factor of
\begin{align}
\text{WF}_{\kpub,\mathrm{distinguisher}} = q^{m\zeta}, \notag
\end{align}
neglecting the cost of checking whether a vector in $\Fqm^n$ is a codeword.
Hence, if $m\zeta$ is smaller than the security parameter of the system, this distinguisher is feasible to compute.

{
\subsection{{Avoiding} Weak Keys}\label{ssec:weak_keys}

As already discussed in Section~\ref{ssec:attacks_ciphertext}, the work factors of the ``Randomized Gabidulin Decoding Attack on the Ciphertext'' and the ``List Decoding of the Ciphertext Attack'' depend on the rank of the error part $\Trmum(\alpha\vec{z}) + \vec{e}$ of the ciphertext (seen as codeword plus error).
Generically, this error has weight $\tpub+\e$, but due to the trace operation and the addition, the rank might be smaller.

In Appendix~\ref{app:probability_of_large_error_weight}, we will analyze the probability that for a given key (i.e., $\z$ in this case) and a random encryption (random choices of $\alpha$ and $\e$) the rank is significantly smaller than expected (we use $<w$ as a threshold, see Section~\ref{ssec:attacks_ciphertext}).
Briefly summarized, we get the following results.

It turns out that this probability heavily depends on the minimum distance of the code $\mathcal{A}$ used to generate $\z$ in Algorithm~\ref{alg:modKeyGen}.
The smaller this minimum distance, the larger the probability that the rank is low.
More precisely, for a given $\mathcal{A}$ of minimum distance $2 \leq t \leq w-\zeta+2$
\begin{equation*}
\Pr(\rank_{\Fq}(\Trmum(\alpha\vec{z}) + \vec{e})<w) \leq q^{-m\zeta}+64 \min\{\trrank,\tpub\}^2 q^{-(\trrank+\tpub-w+1)\left(n+\frac{\trrank-3w-\tpub}{2}\right)},
\end{equation*}
cf.~Theorem~\ref{thm:strong_key_ciphertext_error_weight_probability} in Appendix~\ref{app:probability_of_large_error_weight}.

Due to the above discussion, we call a key with $\Pr(\rank_{\Fq}(\Trmum(\alpha\vec{z}) + \vec{e})<w)>2^{-\secLevel}$ a \emph{weak key}.
In Appendix~\ref{app:probability_of_large_error_weight}, we derive an upper bound on the probability of choosing weak key (i.e., an $\mathcal{A}$ of too small minimum distance) in Algorithm~\ref{alg:modKeyGen}.
For $\zeta q^{\zeta w-m} \leq \tfrac{1}{2}$, this bound is roughly
\begin{equation*}
\Pr(\text{weak key}) \leq \Theta\!\left( q^{m[t-(w-\zeta+2)]} \right),
\end{equation*}
cf.~Remark~\ref{rem:asymptotic_weak_key} (see Theorem~\ref{thm:strong_key_ciphertext_error_weight_probability} for a non-asymptotic bound) in Appendix~\ref{app:probability_of_large_error_weight}, where $t$ is the smallest minimum distance for which the key is not weak.

It can be seen that the parameters of {\SystemName} can be chosen such that there is a $t$ with $2 \leq t \leq w-\zeta+2$ such that both $\Pr(\rank_{\Fq}(\Trmum(\alpha\vec{z}) + \vec{e})<w)$ (for any non-weak key) and $\Pr(\text{weak key})$ are smaller than $2^{-\secLevel}$.
This is the case for all parameters proposed in Table~\ref{tab:parameters}.
}

\subsection{Summary of the Work Factors}

In this section, we recall the conditions on the choice of the parameters such that all known attacks are inefficient and summarize their work factors. Furthermore, we give specific parameters and compare {\SystemName} to other code-based cryptosystems.

In the following, we choose the parameters $q$, $m$, $n$, $k$, $u$, $w$, and $\tpub$ as in Table~\ref{tab:parameters}.
Recall that this choice of $w$ prevents the Overbeck-like attack (Section~\ref{ssec:attack_overbeck-like}) and results in an exponential work factor of the linearization attack (Section~\ref{ssec:attack_linearization}).

Furthermore, we choose $\dimZ$ to be small such that the work factor of searching the exponentially-large output of the interleaved decoding attack (Section~\ref{ssec:attack_interleaved_decoding}) is large.
Note that the latter attack returns an exponentially-large output if and only if the GOT \cite{Gaborit2018-FL-Attack} attack fails, cf.~Theorem~\ref{thm:equiv-int-got}.

The resulting considered work factors are summarized in Table~\ref{tab:wf_attacks}.
In addition to these work factors, we have considered the following requirements:
\begin{itemize}
\item The work factor of the second algebraic attack in \cite{faure2006new} (cf.~Section~\ref{ssec:algebraic_attacks}) is unknown.
Hence, we choose the code parameters such that the resulting non-linear system of equations occurring in the attack consists of more than $n_\mathrm{p} \approx 2^{32}$ many polynomials of degree at least $d_\mathrm{p} = 127$. This is the same choice as in \cite{faure2006new}.
\item Since there is no efficient list decoder for Gabidulin codes, the work factor of the list decoding the public key or the ciphertext in Section~\ref{ssec:list_decoding_attack_ciphertext} is not known.
However, we do have a lower bound on the worst-case work factor for some codes, given by the maximal list size $\mathcal{L}_\text{$\c$,worst}$ in \eqref{eq:list_size_cipher}.
In all examples for which the bound holds, we chose the parameters such that $\log_2(\mathcal{L}_\text{$\c$,worst})$ is much larger than the claimed security level. 
\item {The probability of generating a weak key should be negligible. Thus, we choose the parameters such that $\zeta q^{\zeta w-m} \leq \tfrac{1}{2}$ and
    \begin{align*}
      \Pr(\text{weak key}) &\leq \frac{q^{m\zeta}-1}{(q^m-1)(q^{mw}-1)}\left( \sum_{i=0}^{\trrank-1} \qbinomialField{w}{i}{q} \prod_{j=0}^{i-1} \left(q^m-q^j\right) -1 \right) \\
                             &\leq 2^{-\secLevel},
\end{align*}
where $\secLevel$ is the security parameter and
\begin{align*}
t := \min\left\{ t \, : \,  q^{-m\zeta}+64 \min\{\trrank,\tpub\}^2 q^{-(\trrank+\tpub-w+1)\left(n+\frac{\trrank-3w-\tpub}{2}\right)} \leq 2^{-\secLevel} \right\}.
\end{align*}
}
  
\end{itemize}

\begin{table*}
  \caption{Summary of the Discussed Attacks' Work Factor.}
  \renewcommand{\arraystretch}{1.6} % to make the height of each row equal
\begin{center}
  \begin{tabular}{r|l}
   Name of the attack & Work factor \\
    \hline
    Brute-force $\vec{z}$ (Sec.~\ref{ssec:attack_brute-force_z}) & $\text{WF}_{\text{ILD}}=\frac{q^{w(n-w)}}{\mathcal{N}}$\\
    Interleaved Decoding (Sec.~\ref{ssec:attack_interleaved_decoding}) & $\text{WF}_{\text{ILD}}= \frac{q^{m(w-\dimZ(n-k-w))}}{\mathcal{N}}$  \\
    Randomized Decoding (Sec.~\ref{ssec:attack_rsd_pk}) & $\text{WF}_{\text{RGD}} = \frac{n}{64} q^{m(n-k)-w(n+m)+w^2+\min\{2\xi(\frac{n+k}{2}-\xi),wk\} }$ \\
    Moving to Close Error (Sec.~\ref{ssec:attack_moving_to_error}) &    $\text{WF}_{\text{MCE}} = q^{(n-w)(w(\dimZ+1)-\dimZ(n-k))}$ \\
    \hline
    Combinatorial RSD (Sec.~\ref{ssec:attack_syndrome_decoding}) &$ \text{WF}_{\text{CRSD}}=(n-k)^3m^3q^{\tpub \big \lceil\frac{(k+1)m}{n} \big\rceil -m}$ \\
    Algebraic RSD (Sec.~\ref{ssec:attack_algebraic_syndrome_decoding}) & $ \text{WF}_{\text{ARSD}}$ \\
    Linearization (Sec.~\ref{ssec:attack_linearization}) & $\text{WF}_{\text{Lin}} = q^{m \cdot \max\{u \tpub + u + 1 - \varphi,(u+2)\tpub + k +1 -n\}}$ \\
    GCD based attack (Sec.~\ref{ssec:algebraic_attacks}) & $\text{WF}_{\text{GCD}} = q^{m(u-1)}\frac{q^{\tpub+1}-1}{q-1}$  \\
    Brute-force $\alpha$ (Sec.~\ref{ssec:attack_brute-force-alpha}) & $\text{WF}_{\alpha}=  q^{mu}$ \\
	\hline
	Distinguisher for $\kpub$ (Sec.~\ref{ssec:sec_analysis_IG-Dec}) & $\text{WF}_{\kpub,\mathrm{distinguisher}} = q^{m\zeta}$
  \end{tabular}
\end{center}
\label{tab:wf_attacks}
\end{table*}

\renewcommand{\arraystretch}{1.2}
\setlength{\tabcolsep}{10pt}
\begin{table*}[!t]
	\caption{Parameter sets for 128, 192 and 256 bit security.} 
	\begin{center}
		\begin{tabular}{l||c|c|c|c|c|c|c|c|c }
			Parameter Set &$q$ & $u$ & $k$ & $n$ & $m$ & $\zeta$ & $w$ & $\tpub$ & $R$ \\
			\hline  \hline

                  {\SystemName}-128 & 2 & 5 & 53 & 92 & 92 & 2 & 27  & 6 & 0.52 \\ 

                  {\SystemName}-192 & 2 & 5 & 69 & 120 & 120 & 2 & 35  & 8 & 0.53\\ 

                  {\SystemName}-256 & 2 & 5 & 85 & 148 & 148 & 2 & 43  & 10 & 0.54 \\
                  \hline
                \end{tabular}
	\end{center}
	\label{tab:security_para}
	\vspace*{4pt}
\end{table*}

\begin{table*}[!t]
	\caption{Comparison of memory costs of $\sk$, $\pk$ and the ciphertext $\ct$ in Byte with IND-CCA-secure Loidreau~\cite{Bellini2019indcca} and the NIST proposals RQC~\cite{melchor2019rqc}, ROLLO~\cite{melchor2019rollo}, BIKE~\cite{aragon2019bike} and Classic McEliece~\cite{bernstein2019mceliece}. The entry 'yes' in the column DFR indicates that a scheme has a decryption failure rate larger than 0.} 	\begin{center}
		\begin{tabular}{l||c|c|c|c|c }
			System name & $\sk$ & $\pk$ & $\ct$ & Security & DFR \\
			\hline  \hline
                  \SystemName-128 & 3795 & 6348 & 1058 & 128 bit & no  \\
                  RQC-I & 40 & 1834 & 3652 & 128 bit & no \\
                  ROLLO-I-128 & 40 & 696 & 696 & 128 bit & yes \\
                  Loidreau-128 & 7181 & 6720 & 464 & 128 bit & no \\
                  BIKE-2 Level 1 & 249 & 1271 & 1271 &  128 bit & yes \\
                  McEliece348864 & 6452 & 261120 & 128 & 128 bit & no \\

                  \hline
                  \SystemName-192 & 6450 & 10800 & 1800 & 192 bit & no  \\
                  RQC-II & 40 & 2853 & 5690 & 192 bit & no \\
                  ROLLO-I-192 & 40 & 958 & 958 & 192 bit & yes \\
                  Loidreau-192 & 13548 & 11520 & 744 & 128 bit & no \\
                  BIKE-2 Level 2 & 387 & 2482 & 2482 & 192 bit & yes \\
                  McEliece460896 & 13568 & 524160 & 188 & 192 bit & no \\
                  
                  \hline
                  \SystemName-256 & 9805 & 16428 & 2738 & 256 bit & no  \\
                  RQC-III & 40 & 4090 & 8164 & 256 bit & no \\ 
                  ROLLO-I-256 & 40 & 1371 & 1371 & 256 bit & yes \\
                  Loidreau-256 & 14128 & 16128 & 1024 & 256 bit & no \\
                  BIKE-2 Level 3 & 513 & 4094 &  4094 & 256 bit & yes\\
                  McEliece6688128 & 13892 & 1044992 & 240 & 256 bit & no  \\
				\hline
                  
		\end{tabular}
	\end{center}
	\label{tab:comparison}
\end{table*}

\section{Parameters and Key Sizes}\label{sec:parameters}
We propose parameters for security levels of $128$ bit, $192$ bit and $256$ bit in Table~\ref{tab:security_para}, where $R=\frac{k-u}{n}$ denotes the rate. The parameters are chosen in a way that we can send at least $256$ bit of information and thus the system can be used as a KEM. Further, we use a security margin of at least $20$ bit.  For all parameters, the algebraic attack based on computing gcds of polynomials is the most efficient attack.

To evaluate the performance of \SystemName, we compare it to the IND-CCA-secure version~\cite{Bellini2019indcca} of Loidreau's system~\cite{Loidreau2017-NewRankMetricBased} and the NIST proposals RQC~\cite{melchor2019rqc}, ROLLO~\cite{melchor2019rollo}, BIKE~\cite{aragon2019bike} and Classic McEliece~\cite{bernstein2019mceliece}. We show the sizes of the private key $\sk$, the public key $\pk$ and the ciphertext $\ct$ in Byte in Table~\ref{tab:comparison}.

\section{Conclusion}\label{sec:conclusion}

In this paper, we presented a new rank-metric code-based cryptosystem: {\SystemName}.
{\SystemName} uses a new coding-theoretic interpretation of the Faure--Loidreau system.
We showed that the ciphertext is a corrupted codeword of a Gabidulin code, where to an unauthorized receiver, the error weight is too large to be correctable.
The authorized user knows the row space of a part of the error and is thus able to correct the error.
Further, we derived that a part of the public key can be seen as a corrupted codeword of an interleaved Gabidulin code and that in the original FL system, an interleaved Gabidulin decoder can efficiently recover the private key from this part of the public key with high probability.
We proved that the condition that interleaved Gabidulin decoders fail is equal to the condition that the severe attack by Gaborit, Otmani and Tal\'e Kalachi fails.

Based on the latter observation, we chose {\SystemName}'s key generation algorithm such that interleaved Gabidulin decoders fail which in turn implies that the attack by Gaborit~\emph{et~al.} fails.

We proposed two versions of {\SystemName} and proved that the public key encryption is IND-CPA secure and the KEM is IND-CCA2 secure under the assumption that the \textsf{ResG-Dec} problem is hard.
We extensively analyzed the security of this decisional problem by studying attacks on the \textsf{ResG-Search}, \textsf{ResIG-Search}, and \textsf{ResIG-Dec} (recall that there is a reduction of \textsf{ResG-Dec} to each of the two search problems).
All studied attacks have an exponential work factor in the proposed parameter ranges and can be avoided by parameter choice.

Finally, we presented parameters for security levels of $128$, $192$ and $256$ bit and compared them to the NIST proposals RQC, ROLLO, BIKE, Classic McEliece and a rank-metric McEliece-like system proposed by Loidreau.
It was observed that {\SystemName} has small ciphertext sizes as well as relatively small key sizes.
Encryption and decryption correspond to encoding and decoding of Gabidulin codes, for which efficient and constant-time algorithms exist.
Further, the proposed system guarantees decryption and is not based on hiding the structure of a code.
Hence, the {\SystemName} system should be considered as an alternative of small ciphertext and key size.

\section*{Acknowledgment}
The work of J.~Renner and A.~Wachter-Zeh was supported by the European Research Council (ERC) under the European Union’s Horizon 2020 research and innovation programme (grant agreement no.~801434).

S.~Puchinger received funding from the European Union's Horizon 2020 research and innovation program under the Marie Sklodowska-Curie grant agreement no.~713683.

We would like to thank Johan Rosenkilde for proposing the ``moving to a close error'' attack.
Also, we are thankful to Michael Schelling for his observation that decryption of the FL system can be seen as error-erasure decoding.
Further, we thank Pierre Loidreau for his valuable comments on a previous version of this paper.
{We are also grateful to Alessandro Neri for fruitful discussions that helped to achieve the results in Appendix~\ref{app:probability_of_large_error_weight}.}

\bibliographystyle{splncs04}
\bibliography{main}

\appendix

\section{Practical Considerations on the Key Generation}\label{app:practical_key_generation}

We discuss practical aspects related to the following lines of the modified key generation algorithm (Algorithm~\ref{alg:modKeyGen}).
{\RestyleAlgo{plain}
\begin{algorithm}
\DontPrintSemicolon
{
\setcounter{AlgoLine}{2}
$\mathcal{A} \sample \left\{ \text{subspace } \mathcal{U} \subseteq \Fqm^w \, : \, \dim \mathcal{U} = \zeta, \, \mathcal{U} \text{ has a basis of full-$\Fq$-rank elements} \right\}$ \\
\setcounter{AlgoLine}{2}
\SetNlSty{textbf}{}{'}
$\begin{pmatrix}
\s_1 \\
\vdots \\
\s_u
\end{pmatrix} \sample \left\{ \begin{pmatrix}
\s_1' \\
\vdots \\
\s_u'
\end{pmatrix} \, : \, \langle\s_1',\dots,\s_u'\rangle_{\Fqm} = \mathcal{A}, \, \rank_{q}(\s_i') = w \, \forall \, i \right\}$ \\
}
\label{alg:s_construction}
\end{algorithm}

We conjecture that the set from which $\mathcal{A}$ is sampled is almost the entire set of $\zeta$-dimensional subspaces of $\Fqm^w$ (or, equivalently, of linear $[w,\zeta]_{q^m}$ codes). Using a combinatorial argument on the known number of full-rank codewords of MRD codes, we prove in Lemma~\ref{lem:MRD_basis_of_full_rank_codewords} (Appendix~\ref{app:probability_of_large_error_weight}) that MRD codes always have a basis consisting of full-rank codewords.
Since the weight enumerator is not known in general for non-MRD codes, we cannot give a proof, but we expect that most codes that are close to MRD (i.e., $d$ is close to $n-k+1$) also have such a basis.
The conjecture is then implied by the fact that (close-to) MRD codes constitute the majority of linear codes \cite{neri2018MRD,byrne2020partition} for the parameters considered here.}

Since it is hard to check if a randomly drawn code admits a basis of full-$\Fq$-rank codewords in the worst case, these arguments also imply a practical method on how to implement Lines~$3$ and $3'$ in practice: sample uniformly at random from the set of $[w,\zeta]_{q^m}$ codes. With overwhelming probability, the code is close to MRD and a large proportion of its codewords have full $\Fq$-rank.
Randomly choosing $u$ codewords will thus give a generating set consisting of full-rank codewords with high probability.
Only if no basis is found after a given number of trials, one needs to formally check if the code does not admit a generating set of full-$\Fq$-rank codewords.
This gives a Las-Vegas-type algorithm with (supposedly) small expected running time.

The worst case of this algorithm (i.e., no suitable generating set is found after a given number of trials) occurs with extremely small probability (provably it is close to the probability of drawing no MRD code at random, it might be even smaller in reality since also ``near-MRD'' might have suitable bases).
Nevertheless, the worst-case complexity is still quite large.
Alternatively, one can draw a new code $\mathcal{A}$ if no generating set is found after a given number of trials.
This, however, slightly changes the random experiment from which the code $\mathcal{A}$ is drawn.
The only part of this paper which is influenced by such a modification is Section~\ref{ssec:weak_keys}, a summary of Appendix~\ref{app:probability_of_large_error_weight}, which studies weak keys (i.e., keys for which there is a non-negligible probability that the error part of the ciphertext has too low rank and is vulnerable to a feasible ciphertext attack).
A key is weak only if the minimum distance of $\mathcal{A}$ is small.
By parameter choice, the probability that such a key is generated can be made arbitrarily small (cf.~Appendix~\ref{app:probability_of_large_error_weight}).
By the same arguments as above, we conjecture that if the probability of obtaining a generating set of full-$\Fq$-rank codewords by drawing $u$ codewords uniformly at random is small, then also the minimum distance of the code must be small (i.e., far away from MRD).
In summary, we expect (but cannot prove) that this change of drawing procedure results in an even smaller weak-key probability than predicted by Theorem~\ref{thm:strong_key_ciphertext_error_weight_probability} (Appendix~\ref{app:probability_of_large_error_weight}).

\section{Decryption as Error-Erasure Decoding}\label{app:decryption_as_error-erasure_decoding}

In the following, we give a coding-theoretic interpretation of the ciphertext of the original FL system and of {\SystemName}, which---to the best of our knowledge---has not been observed before.
\begin{lemma}\label{lem:ciphertext_decomposition_decryption_error_erasure}
Fix a basis $\vec{\gamma}$ of $\Fqm$ over $\Fq$. Then, the matrix representation of the ciphertext can be written in the form
\begin{equation}
  \extsmallfield_{\vec{\gamma}}(\c) = \vec{C}_\mathcal{G}  + \vec{A}^{(C)} \vec{B}^{(C)} + \vec{E}, \label{eq:ciphertext_decomposition_decryption_error_erasure}
\end{equation}
where
\begin{itemize}
	\item 
$\vec{C}_\mathcal{G} = \extsmallfield_{\vec{\gamma}}([\vec{m}+\Trmum(\alpha \vec{x})]\cdot\mathbf G_{\mycode{G}}) \in \Fq^{m \times n}$
is unknown and a codeword of a Gabidulin code,
 \item $\vec{A}^{(C)} = \extsmallfield_{\vec{\gamma}}(\Trmum(\alpha \s)) \in \Fq^{m \times w}$
is unknown,
\item  $\vec{B}^{(C)} = (\vec{P}^{-1})_{[1,\dots,w]} \in \Fq^{w \times n}$
is known and
\item $\vec{E} 	  =  \extsmallfield_{\vec{\gamma}}(\vec{e}) \in \Fq^{m \times n}$
is unknown.
\end{itemize}
\end{lemma}

\begin{proof}
  Due to the $\Fqm$-linearity of the trace map $\Trmum$ and the fact that the entries of the matrices $\GG$ and $\Pinv$ are in $\Fqm$, we can write the ciphertext as follows.
  \begin{small}
\begin{align*}
\c &=  \m \GG + \Trmum(\alpha \kpub) + \e \\
&= \m \GG + \Trmum(\alpha \x \GG + \alpha \z) + \e \\
&= \left[ \m + \Trmum(\alpha \x) \right] \GG + \Trmum(\alpha \z) + \e \\
&= \left[ \m + \Trmum(\alpha \x) \right] \GG + \Trmum(\alpha (\s \mid \vec{0}) \Pinv) + \e \\
&= \left[ \m + \Trmum(\alpha \x) \right] \GG + \Trmum(\alpha \s (\vec{P}^{-1})_{[1,w]}) + \e \\
&= \left[ \m + \Trmum(\alpha \x) \right] \GG + \Trmum(\alpha \s) (\vec{P}^{-1})_{[1,w]} + \e.
\end{align*}
\end{small}
Since the entries of $(\vec{P}^{-1})_{[1,\dots,w]}$ are in $\Fq$, the expansion of the ciphertext into the $\Fq$-basis $\vec{\gamma}$ of $\Fqm$ can be written as in~\eqref{eq:ciphertext_decomposition_decryption_error_erasure} above. 
\qed
\end{proof}

\begin{theorem}\label{thm:decryption_as_error_erasure_decoding}
The message vector $\vec{m}$ can be reconstructed by the error-erasure decoders in \cite{GabidulinPilipchuck_ErrorErasureRankCodes_2008,silva_rank_metric_approach,WachterzehZeh-ListUniqueErrorErasureInterpolationInterleavedGabidulin_DCC2014} (as well as their accelerations in \cite{PuchingerWachterzeh-ISIT2016,puchinger2018fast}) and Steps~\ref{line:decrypt_alpha} and~\ref{line:decrypt_m} of Algorithm~\ref{alg:Decrypt}.
\end{theorem}

\begin{proof}
As seen in Lemma~\ref{lem:ciphertext_decomposition_decryption_error_erasure}, we can decompose the matrix representation of the ciphertext into a codeword plus an error that is partially known.
In fact, the decomposition is of the form as in \eqref{eq:decomp_errrorerasures} (see Section~\ref{ssec:preliminaries_rank-metric_codes}), so $\vec{m}+\Trmum(\alpha \vec{x})$ can be reconstructed by the error-erasure decoders in \cite{GabidulinPilipchuck_ErrorErasureRankCodes_2008,silva_rank_metric_approach,WachterzehZeh-ListUniqueErrorErasureInterpolationInterleavedGabidulin_DCC2014} since the decoding condition \eqref{eq:errorerasurecond} reads as
\begin{align*}
w+2 \rank_q(\vec{E}) = w+2\tpub \leq n-k
\end{align*}
in this case and is fulfilled by Table~\ref{tab:parameters}.

The message $\m$ can then be recovered from $\vec{m}+\Trmum(\alpha \vec{x})$ using the same steps as in Algorithm~\ref{alg:Decrypt}.
\qed
\end{proof}

Theorem~\ref{thm:decryption_as_error_erasure_decoding} leads to the following observation.
The ciphertext is a codeword plus an error of rank weight $w+\tpub$, which is beyond the unique decoding radius.
The legitimate receiver can only decrypt since she knows the ($w$-dimensional) row space of a part of the error.
Although the attacker knows the code, she cannot recover the message since she has no further knowledge about the structure of the error.
Note the difference to the code-based McEliece cryptosystem, where the security relies on the fact that an attacker does not know the structure of the code.
We will turn this observation into an exponential-time message attack in Section~\ref{ssec:list_decoding_attack_ciphertext}, which we will consider in our parameter choice.

Furthermore, the procedure implied by Theorem~\ref{thm:decryption_as_error_erasure_decoding} might have a practical advantage compared to the original decryption algorithm.
The code $\mycode{\mathcal{G}}'$ used for decoding in Algorithm~\ref{alg:Decrypt} depends on the private key.
In Theorem~\ref{thm:decryption_as_error_erasure_decoding}, the code is given by $\vec{g}$, which is public and in fact does not need to be chosen randomly in the key generation.\footnote{Note that we described the key generation as in \cite{faure2006new}, where $\vec{g}$ is chosen at random, but this is not necessary for the security of the system.}
Depending on the used algorithm and type of implementation (e.g., in hardware), it can be advantageous in terms of complexity or implementation size if the code is fixed.

\section{Probability of Large Enough Ciphertext Error Weight}\label{app:probability_of_large_error_weight}

In this section, we analyze the probability that the error part $\Trmum(\alpha\vec{z}) + \vec{e}$ of the ciphertext
\begin{align*}
\vec{c} &= (\vec{m},\vec{0}_u) \cdot \mathbf G_{\mycode{G}} + \Trmum(\alpha \kpub) + \vec{e} \\
&= \underbrace{\left[(\vec{m},\vec{0}_u) + \Trmum(\alpha \vec{x})\right] \cdot \mathbf G_{\mycode{G}}}_{\text{codeword}} + \underbrace{\Trmum(\alpha\vec{z}) + \vec{e}}_{\text{error}}
\end{align*}
has large enough rank to avoid the ciphertext attacks discussed in Section~\ref{sec:sec_analysis}.
{The results of this appendix are summarized in Section~\ref{ssec:weak_keys}.}

Generically (i.e., with probability close to $1$ for random choices of $\kpub$, $\alpha$, $\e$), we have $\rank_{q}(\Trmum(\alpha\vec{z}))=w$, $\rank_{q}(\vec{e}) = \tpub$, and $\rank_{q}\big(\Trmum(\alpha\vec{z})+\e\big)=w+\tpub$.
However, there is a very small probability that the error has significantly smaller rank than the generic case.
Our aim is to design the system parameters such that this probability is sufficiently small, e.g., $2^{-\secLevel}$, to avoid attacks utilizing this behavior.

As we will see in this section, the choice of $\z$ in the public key influences this probability (fixed $\z$, randomness in $\alpha$ and $\e$) significantly.
Since $\z$ is itself drawn using a random experiment during the key generation, we study with which probability this key is ``strong'', i.e., whether the rank of $\Trmum(\alpha\vec{z}) + \vec{e}$ is large with sufficiently high probability (randomness only in $\alpha$ and $\e$).

We start with a lemma that shows that the probability mass function of the $\Fq$-rank of $\Trmum(\alpha \z)$ for uniformly drawn $\alpha$ only depends on the weight distribution of the code spanned by $\a_1,\dots,\a_\zeta$ (the $\Fqm$-linearly independent vectors over $\Fqm$ from which $\z$ is constructed).

\begin{lemma}\label{lem:Tr_alpha_z_rank_rank_distribution_A}
Let $\z$ be constructed from the randomly chosen code $\mathcal{A}[w,\zeta]_{q^m}$ as in Algorithm~\ref{alg:modKeyGen}.
Denote by $A_0,\dots,A_w$ the rank-weight distribution of $\mathcal{A}$.
For $\alpha \in \Fqmu$ chosen uniformly at random, we have
\begin{align*}
\Pr\Big(\rank_{\Fq}\big(\Trmum(\alpha \z) \big) = i \Big) = \frac{A_i}{q^{m\zeta}}.
\end{align*}
\end{lemma}

\begin{proof}
We use the notation ($\z$, $\s$, $\mathcal{A}$, $\P$, and $\S$) from Algorithm~\ref{alg:modKeyGen}.
First observe that $\Trmum(\alpha \z) = \Trmum(\alpha [\s \mid \0]) \P$. Hence,
\begin{align*}
\rank_{\Fq}\!\left(\Trmum(\alpha \z) \right) = \rank_{\Fq}\!\left(\Trmum(\alpha \s) \right).
\end{align*}
We can expand $\alpha \in \Fqmu$ in the dual basis $\gamma_i^\ast$ as $\alpha = \sum_{i=1}^{u} \alpha_i \gamma_i^\ast$. Then,
\begin{align*}
\Trmum(\alpha \s) &= \sum_{i=1}^{u} \alpha_i \s_i = \begin{pmatrix}
\alpha_1 & \dots & \alpha_u
\end{pmatrix}
\cdot
\begin{pmatrix}
\s_1 \\ \vdots \\ \s_u \\
\end{pmatrix} = \begin{pmatrix}
\alpha_1 & \dots & \alpha_u
\end{pmatrix}
\S
\begin{pmatrix}
\a_1 \\ \vdots \\ \a_\zeta
\end{pmatrix},
\end{align*}
where $\a_1,\dots,\a_\zeta$ is a basis of $\mathcal{A}$ and $\S \in \Fqm^{u \times \zeta}$ is a matrix of full rank $\zeta$.
As $\alpha$ is chosen uniformly at random from $\Fqmu$, the $\alpha_i$ are chosen independently and uniformly at random from $\Fqm$. As $\rank_{q^m} \S = \zeta$, this is equivalent to saying that
\begin{align*}
\begin{pmatrix}
\beta_1 & \dots & \beta_\zeta
\end{pmatrix}
:=
\begin{pmatrix}
\alpha_1 & \dots & \alpha_u
\end{pmatrix}
\S
\end{align*}
is chosen uniformly at random from $\Fqm^\zeta$. Hence, we have
\begin{align*}
\Trmum(\alpha \s) = \begin{pmatrix}
\beta_1 & \dots & \beta_\zeta
\end{pmatrix}
\begin{pmatrix}
\a_1 \\ \vdots \\ \a_\zeta
\end{pmatrix},
\end{align*}
i.e., $\Trmum(\alpha \s)$ is a codeword of $\mathcal{A}$, chosen uniformly at random.
This immediately implies the claim.
\qed
\end{proof}

A direct consequence of the lemma above is the following statement.
\begin{corollary}\label{cor:Tr_alpha_z_rank_min_dist_A}
With notation as in Lemma~\ref{lem:Tr_alpha_z_rank_rank_distribution_A}, let $d$ be the minimum rank distance of the code $\mathcal{A}[w,\zeta]_{q^m}$. Then,
\begin{align*}
\Pr\Big(\rank_{\Fq}\big(\Trmum(\alpha \z) \big) < d \Big) = q^{-m\zeta}.
\end{align*}
\end{corollary}

Corollary~\ref{cor:Tr_alpha_z_rank_min_dist_A} shows that we can easily bound the probability that $\Trmum(\alpha \z)$ has small $\Fq$-rank if the code $\mathcal{A}$ (as defined in Lemma~\ref{lem:Tr_alpha_z_rank_rank_distribution_A}) has a large minimum rank distance.
Loosely speaking, if the minimum rank distance of the code is small, we can consider this key to be \emph{weak}, and \emph{strong} otherwise.
Since the code is chosen uniformly at random from the set of $\Fqm$-linear $[w,\zeta]_{\Fqm}$ (cf.~choice of $\a_i$ in Algorithm~\ref{alg:modKeyGen}), we can use the following result from \cite{byrne2020partition} to bound the probability that the key is weak.

\begin{lemma}[{\!\!\cite[Corollary~5.4]{byrne2020partition}}]\label{lem:Tr_alpha_z_min_dist_A_probability}
Let $1 \leq k \leq n$ and $2\leq d \leq n-k+2$. Choose a code $\Code \in \Fqm^n$ uniformly at random from the $\Fqm$-linear codes of parameters $[n,k]_{\Fqm}$. Then,
\begin{align*}
\qbinomialField{n}{k}{q^m}^{-1}\sum_{h=1}^{d-1} &\qbinomialField{d-1}{h}{q^m} \sum_{s=h}^{d-1} \qbinomialField{d-1-h}{s-h}{q^m} \qbinomialField{n-s}{n-k}{q^m} (-1)^{s-h} q^{m\binom{s-h}{2}}\leq \\
&\Pr\Big( \dRmin(\Code) < d \Big) \leq \frac{q^{mk}-1}{(q^m-1)(q^{mn}-1)}\left( \sum_{i=0}^{d-1} \qbinomialField{n}{i}{q} \prod_{j=0}^{i-1} \left(q^m-q^j\right) -1 \right)
\end{align*}
\end{lemma}

Since the code in Lemma~\ref{lem:Tr_alpha_z_min_dist_A_probability} is chosen uniformly at random, it does not exactly match the distribution of the code $\mathcal{A}$ in Algorithm~\ref{alg:modKeyGen}.
Hence, we need the following lemma and theorem to estimate the probability of a small minimum distance in our case.

{
\begin{lemma}\label{lem:MRD_basis_of_full_rank_codewords}
An $\Fqm$-linear MRD code $[n,k]_{q^m}$ has a basis consisting of codewords of $\Fq$-rank $n$.
\end{lemma}

\begin{proof}
We show that the number of full-rank codewords is at least $q^{m(k-1)}$.
Since these codewords are all non-zero, their $\Fqm$-span must have cardinality at least $q^{mk}$ and is hence the entire code.

The weight distribution of an MRD code of length $n$ and minimum distance $d$ can be given by (see~\cite{Gabidulin_TheoryOfCodes_1985}):
\begin{equation}\label{eq:weight-dist}
A_{d+s}=\quadbinom{n}{d+s}\sum\limits_{j=0}^{s} (-1)^{j+s} \quadbinom{d+s}{d+j}q^{(s-j)(s-j-1)/2}(q^{m(j+1)}-1), \ s=0,1,\dots,n-d,
\end{equation}
where $m$ is the order of the extension field, $n \leq m$, and $A_{d+s}$ denotes the number of rank-$(d+s)$ codewords.

We are interested in a lower bound for the number of full-rank codewords, i.e., $s = n-d$. 
The sum in~\eqref{eq:weight-dist} is an alternating sum whose terms get larger, the larger $j$ and therefore can be lower bounded by the case of $j=s$ plus the case of $j=s-1$. That means:
 \begin{align}
A_{d+s} &\geq \quadbinom{n}{d+s} \left( (q^{m(s+1)}-1) - \quadbinom{d+s}{d+s-1}(q^{ms}-1)\right).\nonumber
\end{align}
Hence, for $s=n-d$, we obtain:
 \begin{align}
A_{n} &\geq \quadbinom{n}{n} \left( (q^{m(n-d+1)}-1) - \quadbinom{n}{n-1}(q^{m(n-d)}-1)\right)\nonumber\\
&= q^{mk}-1 - \quadbinom{n}{n-1}(q^{m(k-1)}-1)\nonumber\\
& = q^{mk}-1 - \frac{(q^n-1) q^{m(k-1)}}{q-1}+\frac{q^n-1}{q-1}\nonumber\\
& \geq q^{mk}- \frac{(q^n-1) q^{m(k-1)}}{q-1}\nonumber\\
&\geq q^{m(k-1)} \underbrace{\left(q^m - \frac{q^n-1}{q-1}\right)}_{\geq 1 \text{ since } m \geq n \text{ and } q \geq 2}\nonumber\\
&\geq q^{m(k-1)}.
\end{align}
\end{proof}
}

{
\begin{theorem}
Let $m$, $\zeta$, and $w$ be chosen such that
\begin{equation}
1-\zeta q^{\zeta w-m} \geq \tfrac{1}{2}. \label{eq:MRD_probability_condition}
\end{equation}
Let $\mathcal{A}$ be chosen as in Algorithm~\ref{alg:modKeyGen}, i.e., uniformly at random from the set of linear $[w,\zeta]_{q^m}$ codes that have a basis consisting only of codewords with $\Fq$-rank $w$. Furthermore, let $2 \leq t \leq w-\zeta+2$. Then,
\begin{align*}
\Pr\Big( \dRmin(\mathcal{A}) < t \Big) \leq 2 \frac{q^{m\zeta}-1}{(q^m-1)(q^{mw}-1)}\left( \sum_{i=0}^{t-1} \qbinomialField{w}{i}{q} \prod_{j=0}^{i-1} \left(q^m-q^j\right) -1 \right).
\end{align*}
\end{theorem}

\begin{proof}
We define an alternative random experiment, where a code $\mathcal{A}'$ is chosen uniformly from \emph{all} linear $[w,\zeta]_{q^m}$. The sought probability is then given by the conditional probability
\begin{align*}
\Pr\Big( \dRmin(\mathcal{A}') < t \mid \mathcal{S}\Big),
\end{align*}
where $\mathcal{S}$ is the event that $\mathcal{A}'$ has a basis of maximal-rank codewords.
We derive the result using the relation
\begin{align}
\Pr\Big( \dRmin(\mathcal{A}') < t \Big) \geq \Pr\Big( \dRmin(\mathcal{A}') < t \mid \mathcal{S}\Big)\Pr\Big(\mathcal{S}\Big). \label{eq:probability_of_d_min_small_conditional}
\end{align}
First note that Lemma~\ref{lem:Tr_alpha_z_min_dist_A_probability} gives us
\begin{align*}
\Pr\Big( \dRmin(\mathcal{A}') < t \Big) \leq \frac{q^{m\zeta}-1}{(q^m-1)(q^{mw}-1)}\left( \sum_{i=0}^{t-1} \qbinomialField{w}{i}{q} \prod_{j=0}^{i-1} \left(q^m-q^j\right) -1 \right).
\end{align*}
By Lemma~\ref{lem:MRD_basis_of_full_rank_codewords}, we have
\begin{align*}
\Pr\big(\mathcal{S}\big) \geq \Pr\big( \mathcal{A}' \text{ is MRD} \big).
\end{align*}
Using \cite[Theorem~21]{neri2018MRD}, we can lower-bound this probability by
\begin{align*}
\Pr\big( \mathcal{A}' \text{ is MRD} \big) \geq 1-\zeta q^{\zeta w-m} \geq \tfrac{1}{2},
\end{align*}
where the last inequality follows from \eqref{eq:MRD_probability_condition}.
The claim follows by combining the two bounds with \eqref{eq:probability_of_d_min_small_conditional}. 
\qed
\end{proof}
}

The last building block for a general bound on the probability of $\Trmum(\alpha\vec{z}) + \vec{e}$ having small rank is the following lemma, which gives a bound for this probability conditioned on the event that $\Trmum(\alpha\vec{z})$ has a given (large) rank.

\begin{lemma}\label{lem:key_error_weight_condition_probability}
Let $\kpub = \vec{x} \cdot \mathbf G_{\mycode{G}} + \vec{z}$ be fixed as in Algorithm~\ref{alg:Encrypt} and let $\alpha$ be chosen such that $\rank_{\Fq}\!\left(\Trmum(\alpha\vec{z})\right)=\trrank$. For $\vec{e} \sample \{ \a \in \Fqm^n : \rank_q(\a) = \tpub \}$, drawn uniformly at random, we have
\begin{align*}
\Pr\!\left(\rank_{\Fq}\Big( \Trmum(\alpha\vec{z})+\e\Big)<w \;\Big\vert \rank_{\Fq}\Big( \Trmum(\alpha\vec{z})\Big)=\trrank, \, \rank_{\Fq}(\e)=\tpub \right) \\
\leq 64 \min\{\trrank,\tpub\}^2 q^{-(\trrank+\tpub-w+1)\left(n+\frac{\trrank-3w-\tpub}{2}\right)} \\
\end{align*}
\end{lemma}

\begin{proof}
For simplicity, we write (for some basis ${\boldsymbol{\gamma}}$ of $\Fqm$ over $\Fq$)
\begin{align*}
\e_1 &:= \Trmum(\alpha\vec{z}), &&\E_1 := \extsmallfield_{\boldsymbol{\gamma}}(\e_1), &&\mathcal{E}_1^\mathrm{C} := \mathrm{ColumnSpace}(\E_1) \subseteq \Fq^m \\
& && &&\mathcal{E}_1^\mathrm{R} := \mathrm{RowSpace}(\E_1) \subseteq \Fq^n \\
\e_2 &:= \Trmum(\alpha\vec{z}), &&\E_2 := \extsmallfield_{\boldsymbol{\gamma}}(\e_2), &&\mathcal{E}_2^\mathrm{C} := \mathrm{ColumnSpace}(\E_2) \subseteq \Fq^m \\
& && &&\mathcal{E}_2^\mathrm{R} := \mathrm{RowSpace}(\E_2) \subseteq \Fq^n.
\end{align*}
It is clear that $\rank_{\Fq}(\e_1+\e_2) = \rank_{\Fq}\!\left(\E_1 + \E_2\right)$ and, since $\rank_{\Fq}(\e_1)= \rank_{\Fq}(\E_1) = \trrank$ and $\rank_{\Fq}(\e_1)= \rank_{\Fq}(\E_1) = \tpub$,
\begin{align*}
\dim_{\Fq}\!\left(\mathcal{E}_1^\mathrm{C}\right) &= \dim_{\Fq}\!\left(\mathcal{E}_1^\mathrm{R}\right) = \trrank \\
\dim_{\Fq}\!\left(\mathcal{E}_2^\mathrm{C}\right) &= \dim_{\Fq}\!\left(\mathcal{E}_2^\mathrm{R}\right) = \tpub.
\end{align*}
Note that in our probabilistic model, $\mathcal{E}_1^\mathrm{C}$ and $\mathcal{E}_1^\mathrm{R}$ are fixed and it follows easily that $\mathcal{E}_2^\mathrm{C}$ and $\mathcal{E}_2^\mathrm{R}$ are random variables that are uniformly distributed on the set of $\tpub$-dimensional subspaces of $\Fq^m$ and $\Fq^n$, respectively, and stochastically independent.
Due to \cite[Theorem~1]{marsaglia1967bounds}, for
\begin{equation*}
\dim\!\left(\mathcal{E}_1^\mathrm{C} \cap \mathcal{E}_2^\mathrm{C}\right) = i \quad \text{and} \quad \dim\!\left(\mathcal{E}_1^\mathrm{R} \cap \mathcal{E}_2^\mathrm{R}\right) = j,
\end{equation*}
we have
\begin{align*}
\rank_{\Fq}\!\left(\E_1\right) + \rank_{\Fq}\!\left(\E_2\right) - i - j \leq \rank_{\Fq}\!\left(\E_1 + \E_2\right) \leq \rank_{\Fq}\!\left(\E_1\right) + \rank_{\Fq}\!\left(\E_2\right) - \max\{i,j\}.
\end{align*}
Since $\rank_{\Fq}\!\left(\E_1\right) + \rank_{\Fq}\!\left(\E_2\right) = \trrank+\tpub$, this implies
\begin{align*}
&\Pr\!\left(\rank_{\Fq}\Big( \E_1 + \E_2 \Big)<w \; \Big\vert \rank_{\Fq}\Big(\E_1\Big)=\trrank, \, \rank_{\Fq}\!\left(\E_2\right)=\tpub \right) \\
&\leq \Pr\Big(\dim\!\left(\mathcal{E}_1^\mathrm{C} \cap \mathcal{E}_2^\mathrm{C}\right)+\dim\!\left(\mathcal{E}_1^\mathrm{R} \cap \mathcal{E}_2^\mathrm{R}\right) >\trrank+\tpub-w \Big) \\
&=\sum_{\substack{i,j=0 \\ i+j>\trrank+\tpub-w}}^{\min\{\trrank,\tpub\}} \Pr\Big(\dim\!\left(\mathcal{E}_1^\mathrm{C} \cap \mathcal{E}_2^\mathrm{C}\right) = i \Big) \Pr\Big(\dim\!\left(\mathcal{E}_1^\mathrm{R} \cap \mathcal{E}_2^\mathrm{R}\right) = j \Big).
\end{align*}
Due to \cite[Proof of Lemma~7]{etzion2011error}, we have
\begin{align*}
\Pr\!\left(\dim\!\left(\mathcal{E}_1^\mathrm{C} \cap \mathcal{E}_2^\mathrm{C}\right) = i \right) &= \frac{\qbinomial{m-w}{\tpub-i}\qbinomial{w}{i} q^{(w-i)(\tpub-i)}}{\qbinomial{m}{\tpub}} \\
&\leq 16 \frac{q^{(\tpub-i)(m-w-\tpub+1) + i(w-i)+(w-i)(\tpub-i)}}{q^{\tpub(m-\tpub)}} \\
&= 16 q^{-i(m-w-\tpub+i)}.
\end{align*}
Likewise, we have
\begin{align*}
\Pr\!\left(\dim\!\left(\mathcal{E}_1^\mathrm{R} \cap \mathcal{E}_2^\mathrm{R}\right)= i \right) \leq 16 q^{-i(n-w-\tpub+i)}.
\end{align*}
Due to $n \leq n$, we obtain
\begin{align*}
&\Pr\!\left(\rank_{\Fq}\Big( \E_1 + \E_2 \Big)<w \; \Big\vert \rank_{\Fq}\Big(\E_1\Big)=\trrank, \, \rank_{\Fq}\!\left(\E_2\right)=\tpub \right) \\
&\leq 64 \sum_{\substack{i,j=0 \\ i+j>\trrank+\tpub-w}}^{\min\{\trrank,\tpub\}} q^{-i(n-w-\tpub+i)} q^{-j(m-w-\tpub+j)} \\
&\leq 64 \sum_{\substack{i,j=0 \\ i+j>\trrank+\tpub-w}}^{\min\{\trrank,\tpub\}} q^{-i(n-w-\tpub+i)} q^{-j(n-w-\tpub+j)} \\
&\leq 64 \sum_{\substack{i,j=0 \\ i+j>\trrank+\tpub-w}}^{\min\{\trrank,\tpub\}} q^{-(i+j)(n-w-\tpub)-(i^2+j^2)} \\
&\leq 64 \min\{\trrank,\tpub\}^2 q^{-(\trrank+\tpub-w+1)(n-w-\tpub)-\frac{(\trrank+\tpub-w+1)^2}{2}} \\
&\leq 64 \min\{\trrank,\tpub\}^2 q^{-(\trrank+\tpub-w+1)\left(n+\frac{\trrank-3w-\tpub}{2}\right)}.
\end{align*}
This proves the claim.
\qed
\end{proof}

Summarized, we have the following.
The proof follows directly by combining Corollary~\ref{cor:Tr_alpha_z_rank_min_dist_A}, Lemma~\ref{lem:Tr_alpha_z_min_dist_A_probability}, and Lemma~\ref{lem:key_error_weight_condition_probability}, and a union-bound argument.

\begin{theorem}\label{thm:strong_key_ciphertext_error_weight_probability}
{
Let $m$, $\zeta$, and $w$ be chosen such that $1-\zeta q^{\zeta w-m} \geq \tfrac{1}{2}$.}
Choose $\z$ of the public key as in Algorithm~\ref{alg:modKeyGen}.
Let $2 \leq \trrank \leq w-\zeta+2$.
With probability at least
\begin{align*}
P_\mathrm{strong,key}(\trrank) \geq 1-{2}\frac{q^{m\zeta}-1}{(q^m-1)(q^{mw}-1)}\left( \sum_{i=0}^{\trrank-1} \qbinomialField{w}{i}{q} \prod_{j=0}^{i-1} \left(q^m-q^j\right) -1 \right)
\end{align*}
the public key has the following property:

Choose $\alpha \in \Fqmu$ and For $\vec{e} \sample \{ \a \in \Fqm^n : \rank_q(\a) = \tpub \}$, both uniformly at random. Then the probability that $\Trmum(\alpha \z) + \e$ has $\Fq$-rank at least $w$ is lower-bounded by
\begin{align*}
\Pr\Big(\rank_{\Fq}\big(\Trmum(\alpha \z) &+ \e \big) \geq w \Big) \\
&\geq 1-q^{-m\zeta}-64 \min\{\trrank,\tpub\}^2 q^{-(\trrank+\tpub-w+1)\left(n+\frac{\trrank-3w-\tpub}{2}\right)}.
\end{align*}
\end{theorem}

{
\begin{remark}\label{rem:asymptotic_weak_key}
By the asymptotical analysis in \cite{byrne2020partition}, we have
\begin{align*}
P_\mathrm{strong,key}(\trrank) \geq 1-\Theta\!\left( q^{m[t-(w-\zeta+2)]} \right).
\end{align*}
Since the hidden constant strongly depends on $q$, this asymptotic value should only be used for a rough estimation of the strong-key probability and the exact formula in Theorem~\ref{thm:strong_key_ciphertext_error_weight_probability} should be used for parameter design.

Nevertheless, the formula shows that $1-P_\mathrm{strong,key}(\trrank)$ decreases exponentially in $m$ times the difference of $t$ and $w-\zeta+2$.
Hence, usually we can choose $t$ close to the maximal value $w-\zeta+2$ to achieve a given designed probability for a key to be strong.

For instance, we can choose $t \approx (w-\zeta+2)-\tfrac{\secLevel}{m}\log_{q}(2)$ for
\begin{equation*}
P_\mathrm{strong,key}(\trrank) \geq 1-2^{-\secLevel},
\end{equation*}
where $\secLevel$ is the security parameter.
\end{remark}
}

\end{document}